\documentclass[article]{amsart}
\date{\today}

\usepackage{amsmath,amsthm}
\usepackage{amssymb}
\usepackage{bbm}
\usepackage{mathrsfs}               % \mathscr
\usepackage{enumerate}               % \mathscr
\usepackage{color}

\newtheorem{corollary}{Corollary}
\newtheorem{proposition}{Proposition}
\newtheorem{lemma}{Lemma}
\newtheorem{definition}{Definition}
\newtheorem{theorem}{Theorem}

\newtheorem{remark}{Remark}

\newcommand{\N}{\mathbb{N}} % Natural numbers
\newcommand{\C}{\mathbb{C}} % Complex numbers
\newcommand{\Z}{\mathbb{Z}} % Integer numbers
\newcommand{\R}{\mathbb{R}} % Real numbers

\newcommand{\bx}{{\bf x}}
\newcommand{\by}{{\bf y}}

\newcommand{\HA}{H_{{\bf A}}}
\newcommand{\DA}{D_{{\bf A}}}

\newcommand{\SA}{  {\rm sgn}({\DA})}
\newcommand{\Sps}[2]{\big\langle #1,#2 \big\rangle} 
\newcommand{\sps}[2]{\langle #1,#2 \rangle} %Scalar Product
\newcommand{\core}{C_0^\infty(\R^2,\C^2)}
\newcommand{\hilbert}{L^2(\R^2,\C^2)}
\newcommand{\ri}{\mathrm{i\,}}

\newcommand{\rd}{\mathrm{d}}
\newcommand{\pn}{{k_n}}

\newcommand{\curla}{{\rm curl\,}{\bf A}}
\title[On the essential spectrum of two-dimensional Pauli operators]
{On the essential spectrum of two-dimensional Pauli operators with repulsive potentials}
\author{Josef Mehringer}
\address{Josef Mehringer\\
Fakult\"at f\"ur Mathematik\\
Ludwig-Maximilians-Universit\"at M\"unchen\\
Theresienstra\ss e 39\\
80333 M\"unchen, Germany.}
\email{josef.mehringer@gmail.com}

\begin{document}
\begin{abstract}
We investigate the spectrum of the two-dimensional Pauli operator, describing 
a spin-$\tfrac{1}{2}$ particle in a magnetic field $B$, with a negative scalar
potential $V$ such that $|V|$ grows at infinity. In particular, we obtain
criteria for discrete and dense pure-point spectrum.
\end{abstract}
\maketitle
\section{Introduction}
For modeling the kinetic energy of a non-relativistic
spin-$\tfrac{1}{2}$ particle in the plane, moving
under a magnetic field $B$ in the perpendicular direction 
to the plane, one uses the two-dimensional Pauli operator 
$$ H_{\bf A} :=
\big[{\boldsymbol \sigma}\cdot 
\big(- \ri \nabla- {\bf A}\big) \big]^2 = 
\big(- \ri \nabla- {\bf A}\big)^2 - \sigma_3 B
\quad \mbox{on}
\ \hilbert ,
$$
where ${\bf A}$ is a vector potential associated to $B$,
i.e. $B=\curla :=\partial_1A_2 -\partial_1 A_1$.
Here, ${\boldsymbol \sigma} =(\sigma_1,\sigma_2)$ and
$\sigma_3$ are  the Pauli matrices
\begin{align*}
\sigma_1 = 
\begin{pmatrix}
0 & 1 \\ 1 & 0
\end{pmatrix} , \quad
\sigma_2 = 
\begin{pmatrix}
0 & -\ri \\ \ri & 0
\end{pmatrix}, \quad
\sigma_3 = 
\begin{pmatrix}
1 & 0 \\ 0 & -1
\end{pmatrix}.
\end{align*}
To study the behaviour of such spin-$\tfrac{1}{2}$ particles
(e.g. electrons)
in presence of an additional electric potential $V$, we 
investigate the spectrum of the operator
$$
H := H_{\bf A} + V =
\big[{\boldsymbol \sigma}\cdot 
\big(- \ri \nabla- {\bf A}\big) \big]^2 + V
\quad \mbox{on}
\ \hilbert.
$$
If $V$ is non-negative or decays at infinity 
(e.g. potentials with Coulomb singularities), spectral 
properties of the magnetic Schr\"odinger operator
$\big(- \ri \nabla- {\bf A}\big)^2 + V$,
as well as of the Pauli operator $\HA + V$
(in dimension d = 2 or 3)
have been widely studied over the last decades 
(see, e.g. \cite{Avron_Herbst_Simon_1}, \cite{Kirsch_simon} or
\cite{Erdoes_overview} for a latest overview). 
In this article, instead, we want to point out some 
interesting features of the spectrum of $H$ for potentials $V$ 
tending to $-\infty$ as $|\bx| \to \infty$.  Since such
scalar potentials result in an operator $H$,  unbounded from below,
it is necessary to discuss questions related to the self-adjointness
of $H$. We emphasize that, since we also consider unbounded 
magnetic fields $B$,  the self-adjointness of $H$ cannot simply 
be reduced to the one of the magnetic Schr\"odinger operator.

One motivation for the following considerations is an observation
made in \cite{MehringStock} for the two-dimensional massless
magnetic Dirac operator coupled to an electric potential $V$; 
There, an accumulation process of spectral points has been observed, 
governed by the ratio $|V^2/B|$ at infinity. This phenomenon can be  
ascribed to the non-confining effect of $V$ in the case of the Dirac 
operator. Regarding the Pauli operator, the influence of an additional 
scalar potential $V$ on the spectrum $\sigma(H)$ depends crucially 
on the sign of $V$.  For simplicity, we outline this dependence in the
case of a constant magnetic field $B(x) = B_0$:
a positive potential $V$ growing at infinity
always leads to discrete spectrum of the operator $H$, 
independently of the field strength $B_0$
(see e.g. \cite{kondratiev2005}). Such trapping potentials 
only enhance a localization effect caused by $B$ that 
generates eigenvalues and spectral gaps (proportional to $B_0$). 
If we instead consider negative potentials $V$, the situation is 
quite different since the particle lowers its energy by moving
in regions where $V$ is small.
A scalar potential $V$ converging to $- \infty$ as $|\bx| \to \infty$
has therefore a delocalizing effect, i.e. the particle tends to escape 
any compact region of the  plane. Our results show that such negative
potentials $V$ (describing for example constant radial fields)
counteract the localizing effect of ``hard'' magnetic fields $B$, as 
they close spectral gaps induced by $B$:
\begin{itemize}
\item If $V$ converges to $-\infty$, but remains small compared
          to $B$,  the spectrum  $\sigma(H)$ is discrete, i.e. it
          consists only of eigenvalues of finite multiplicity. 
\item If $V$ is comparable to $B$, more precisely $|V | \approx 2B$ at
          infinity,  points in the essential
          spectrum occur.
\item If $V$ overtakes $B$, more explicitly $|V/B| \to \infty$ as
          $|\bx| \to  \infty $ (at least along a path),  the spectrum 
          $\sigma(H)$ covers the whole real line.
\end{itemize}
One may compare the third claim with the result in \cite{SimonMiller80}
on $\HA$ for decaying magnetic fields.
The precise statements of the claims above are contained in 
Theorems \ref{thm1}$-$\ref{thm4} of Sect. \ref{mainresults}.
We remark that the case $|V/B| \to \infty$ as 
$|\bx| \to \infty$ is treated by Theorems \ref{thm3} and \ref{thm4}.
Unlike Theorem \ref{thm4}, which is only valid for constant 
magnetic fields, Theorem \ref{thm3} covers also non-constant fields $B$,
but requires stronger constraints on the growth of $V$. Thus, 
the important case $B =B_0$ is adressed by two theorems.
The ideas of the proofs of Theorem \ref{thm1}$-$\ref{thm3} originate
from those used to prove the results in \cite{MehringStock}. However,
since we work with a second-order operator, the proofs are technically 
more laborious. Theorem \ref{thm4} is based on a further
construction of a Weyl sequence, obtained by treating $V$
locally as a potential of a constant electric field. This is a refined 
ansatz compared to the method used
for the proof of Theorem \ref{thm3}.\\ 
%%%%%%%%%%%%%%%%%%%%%%%%%%%%%%Organization%%%%%%%%%%%%%%%%%%%%%%%%%%%%%%%%%%%%%%%%

{\it The organization of this article is as follows:} In the next
section some known facts about the Pauli operator are recapitulated.
We present our precise results in Section \ref{mainresults},
provided with some remarks and important applications.
In Section \ref{prooft1} we give the proof of Theorem \ref{thm1}.
The proofs of Theorems \ref{thm2} and \ref{thm3} are contained in 
Section \ref{prooft2}, while the proof of Theorem \ref{thm4}
can be found in the last section.
In the appendix, attached to the main text, we give a proof
of the essential self-adjointness of the Pauli operator.
%%%%%%%%%%%%%%%%%%%%%%%%%%%%%%%%%%%%%%%%%%%%%%%%%%%%%%%%%%%%%%%%%%%%%%%%%%%%%%%
%%%%%%%%%%%%%%%%%%%%%%%%%%%%%%%%%%%%%%%%%%%%%%%%%%%%%%%%%%%%%%%%%%%%%%%%%%%%%%%
%%%%%%%%%%%%%%%%%%%%%%%%%%%%%%%%%%%%%%%%%%%%%%%%%%%%%%%%%%%%%%%%%%%%%%%%%%%%%%%
%%%%%%%%%%%%%%%%%%%%%%%%%%%%%%%%%%%%%%%%%%%%%%%%%%%%%%%%%%%%%%%%%%%%%%%%%%%%%%%
%%%%%%%%%%%%%%%%%%%%%%%%%%%%%%%%%%%%%%%%%%%%%%%%%%%%%%%%%%%%%%%%%%%%%%%%%%%%%%%
\section{Basic properties of the Pauli operator}\label{susy}
In this section we point out some basic facts about the Pauli operator
and the massless Dirac operator $\DA$, whose square equals $\HA$.
For a vector potential ${\bf A} \in C^1(\R^2,\R^2)$ generating the field
$B =\curla \in C(\R^2, \R)$,
the Hamiltonian $\DA$ is defined as the closure of the operator
\begin{align}
  \label{eq:5}
 {\boldsymbol \sigma}\cdot \big(- \ri \nabla- {\bf
    A}\big)=\left(
\begin{array}{cc}
0&d^*\\
d&0  
  \end{array} \right)
\quad 
\mbox{on} \ \core,
\end{align}
which is essentially self-adjoint on the given core (see
\cite{Chernoff77}). In particular, $d$ and $d^*$ can be 
seen as closed operators, i.e. we use the notation
\begin{align}\label{carolina}
d =\overline{-\ri \partial_1 -A_1+
\ri(-\ri \partial_2 -A_2)
\upharpoonright}_{C^\infty_0(\R^2,\C)} 
\end{align}
and analogously for $d^*$. 
One observes that $d, d^*$ satisfy the commutation relation 
\begin{align}
  \label{eq:1}
  [d,d^*]\varphi:=(dd^*-d^*d)\varphi=2B\varphi
\quad {\rm for}\ \varphi \in C^\infty_0(\R^2,\C).
\end{align}
We can write
$$\HA = \DA ^2 = \left(
\begin{array}{cc}
d^*d&0\\
0&dd^* 
 \end{array}
 \right)
\quad 
\mbox{on} \ \core  $$
and consider $\HA$ as a self-adjoint operator on 
$\{ \psi \in {\mathcal D}(\DA) \, | \, \DA \psi \in {\mathcal D} (\DA) \}$ 
given by the Friedrichs extension.
%%%%%%%%%%%%%%%%%%%%%%%%%%%%%%%%%%%%%%%%%%%%%%%%%%%%%%%%%%%%%%%%%%%%%%%%%%%%%%
The two components $dd^*$ and $d^*d$ of $\HA$ are unitarily 
equivalent on the orthogonal complement of
${\rm ker}(\HA) ={\rm ker}(\DA)$.
To verify this we note first that, due to
the matrix structure of $\DA$, we have
\begin{align}
  \label{eq:8}
  {\rm sgn}({\DA}):=\frac{\DA}{|\DA|}=\left(\begin{array}{cc}
0&s^*\\
s&0  
\end{array}\right)
\end{align}
on
${\rm ker} (\DA)^\perp=
{\rm ker}(d)^\perp\oplus {\rm ker}(d^*)^\perp$.
Since $ {\rm sgn}({\DA})^2 = {\rm Id}$ on 
${\rm ker}(\DA)^\perp$, the maps 
\begin{align}\label{super0}
  &s:{\rm ker}(d)^\perp\to{\rm ker}(d^*)^\perp, \qquad 
s^*: {\rm ker}(d^*)^\perp\to {\rm ker}(d)^\perp
\end{align}
are unitary and conjugated to each other. By the operator 
identity $\HA = \DA^2 = \SA \DA^2 \SA $ one concludes that 
\begin{align}
\label{super}
\left(
 \begin{array}{cc}
d^*d&0\\
0&dd^*  
  \end{array}
  \right)\varphi = \left(\begin{array}{cc}
s^*dd^*s&0\\
0&sd^*ds^*  
  \end{array}
  \right) \varphi 
\end{align}
for any 
$\varphi=(\varphi_1,\varphi_2)^{\rm T}$
with $\varphi_1\in \mathcal{D}({d^*d})\cap {\rm ker}(d)^\perp$ and
$\varphi_2\in \mathcal{D}({dd^*})\cap {\rm ker}(d^*)^\perp$. Hence,
on $\ker (d)^\perp$ the operator $d^*d$ is unitarily
equivalent to $dd^*$ (considered as an operator
on $\ker (d^*)^\perp$). Let us denote the orthogonal 
projection on ${\rm ker} (\DA)$ by $P_0$ and the orthogonal 
projections on ${\rm ker} (d)$, ${\rm ker} (d^*)$ by  $\pi, \pi_*$.  
Further, we set
$$ P_0^\perp := \mathbbm{1}- P_0 , \quad
\pi^\perp := \mathbbm{1}- \pi , \quad
\pi_*^\perp := \mathbbm{1}- \pi_* . $$

%%%%%%%%%%%%%%%%%%%%%%%%%%%%%%%%%%%%%%%%%%%%%%%%%%%%%%%%%%%%%%%%%%%%%%%%%%%%%%%
To define our full Hamiltonian let 
${\bf A} \in C^1(\R^2,\R^2)$ and $B,V \in C(\R^2,\R)$ 
be such that $B=\curla$, then $H$ is given by 
$$H \varphi = \big[\DA ^2 + V\big] \varphi  =  
\big[ \big(-\ri \nabla -{\bf A}\big)^2 - \sigma_3 B + V \big]
\varphi \quad 
\mbox{for} \ \varphi \in \core\,. $$
In general, the closure of this densely defined operator is not 
self-adjoint without any restriction on the growth rate of $V$
at infinity. However, there are conditions, very similar to those
for the classical Schr\"odinger operator, to ensure essential
self-adjointness.
\begin{proposition}\label{EssSelf}
Let $B, V \in C^1(\R^2,\R)$
and ${\bf A}\in C^2(\R^2,\R^2)$ 
with $B=\curla$. In addition, 
assume that $V$ fulfills the lower bound
\begin{align}\label{growthcond}
V(\bx) \ \ge \ -c|\bx|^2 + d\,, \quad  \ \bx \in \R^2,
\end{align}
for some constants $c>0$, $d \in \R$.
Then $H$ is essentially self-adjoint on $\core$. 
 \end{proposition}
%%%%%%%%%%%%%%%%%%%%%%
\begin{remark}\label{relaxregularity}
Following the lines of the proof given in Appendix \ref{selfadj},
we see that the regularity condition on 
$B, V$ can be relaxed to $B,V \in C^\alpha_{loc}(\R^2,\R)$, i.e.
both only need to be locally $\alpha$-H\"older continuous.
By a perturbation argument one can also see that it suffices
to assume that $B,V$ are $C^\alpha_{loc}$ outside some compact set
$K\subset \R^2$, while inside $K$ they only need to be continuous.
\end{remark}
%%%%%%%%%%%%%%%%%%%%%%
\begin{remark}
The self-adjoint operator given by Proposition \ref{EssSelf} 
is locally compact, i.e. for any characteristic function
$\chi_{B_R(0)}$ on the ball $B_R(0)$ with radius $R$, the operator
$\chi_{B_R(0)}  (H-\ri)^{-1} $
is compact.
\end{remark}
%%%%%%%%%%%%%%%%%%%%%%%%%%%%%%%%%%%%%%%%%%%%%%%%%%%%%%%%%%%%%%%%%%%%%
\begin{remark}
Considering the case $V=0$,  we obtain that $\HA$,
$dd^*$ and $d^*d$ are essentially self-adjoint on 
$\core$, respectively on $C_0^\infty(\R^2,\C)$.
\end{remark}
%%%%%%%%%%%%%%%%%%%%%%%%%%%%%%%%%%%%%%%%%%%%%%%%%%%%%%%%%%%%%%%%%%%%%
Note that \eqref{growthcond} is the same lower bound on $V$ as one needs
for the (magnetic) Schr\"odinger operator to ensure the essential
self-adjointness, whereas no restriction on the growth of $B$ is necessary.
The regularity conditions on $V$ and ${\bf A}$ are quite strong compared
to those for the magnetic Schr\"odinger operator (see \cite{Reed_Simon_2}). 
The reason is that due to the lack of a diamagnetic inequality for $\HA$,
one uses a direct argument that requires more regularity on the potentials 
${\bf A}$ and $V$. The interesting question remains: Could one 
relax these conditions for the Pauli operator as in the case of the 
magnetic Schr\"odinger operator?
%%%%%%%%%%%%%%%%%%%%%%%%%%%%%%%%%%%%%%%%%%%%%%%%%%%%%%%%%%%%%%%%%%%%%%%%%%
%%%%%%%%%%%%%%%%%%%%%%%%%%%%%%%%%%%%%%%%%%%%%%%%%%%%%%%%%%%%%%%%%%%%%%%%%%
%%%%%%%%%%%%%%%%%%%%%%%%%%%%%%%%%%%%%%%%%%%%%%%%%%%%%%%%%%%%%%%%%%%%%%%%%%
\section{Main results}\label{mainresults}
In this section we assume that $B,V$ and ${\bf A}$ satisfy 
the conditions of Proposition \ref{EssSelf}. It is easy to see
that in the following results 
$B,V \in C^1(\R^2,\R)$ can be relaxed to hold
only outside some compact set $K \subset \R^2$ 
as in Remark \ref{relaxregularity}.
\begin{theorem}\label{thm1}
Assume that
\begin{eqnarray}
&&V(\bx)\longrightarrow
 -\infty\quad\mbox{as}\quad|\bx|\to \infty,\label{t1}\\
&&\left|\frac{\nabla V (\bx)}{V(\bx)}\right| \,
       \longrightarrow 0\quad\mbox{as}\quad|\bx| \to \infty,\label{t2}\\
&&\limsup_{|\bx|\to \infty}
     \left| \frac{V(\bx)}{2B(\bx)} \right| <1. \label{t3}
\end{eqnarray}
Then $\sigma_{\rm ess}(H)=\emptyset$, i.e.
$H$ has purely discrete spectrum.
\end{theorem}
%%%%%%%%%%%%%%%%%%%%%%%%%%%%%%%%%%%%%%%%%%%%%%%%%%%%%%%%%%%%%%%%%%%%%%%%%%%
Condition \eqref{t2} is a restriction on the growth rate of $V$
and rather of technical necessity. 
The interplay between $B$ and $V$ (as mentioned in the
introduction)  is described by condition \eqref{t3}.
Thus, it is worthwhile to investigate further the dependence 
of $\sigma(H)$ on this quotient: 

One can easily observe that if the quotient of \eqref{t3} surpasses
the constant $1$,  the spectrum of $H$ changes its character.
To see  this pick $\Omega \in {\rm ker}(d^*d)$, then 
$$(dd^*+V)\Omega = (d^*d+ 2B +V)\Omega \approx 0$$
if $2B \approx -V$. Therefore, if ${\rm ker}(d^*d)$ contains
enough functions 
(which is the case for fields $B$ bounded form below by some 
positive constant), we obtain points in the essential 
spectrum of $H$. 
One can even show that the condition $2B \approx -V$ (at infinity)
does not need to hold globally for obtaining  
$\sigma_{\rm ess}(H)\neq \emptyset$. We demonstrate this for a certain 
class of fields $B$ and potentials $V$.  
\begin{definition}
A function $f:\R^2\to\R$ varies with rate $\nu \in [0,1]$ on a 
set $X \subset \R^2$ if there is a constant $C>0$ such that 
for all $\bx \in X$ it holds that 
\begin{align*}
  |f(\bx +\by)|\le C |f(\bx)|, 
\end{align*}
whenever $\by \in \R^2$ satisfies $|\by| \le \tfrac{1}{2}|\bx|^\nu$
\end{definition}
Note that functions of the form $f_1(\bx) = c|\bx|^s$ and
$f_2(\bx)= c|x_1|^s$, with $c,s \in \R$, vary with any rate 
$\nu \in [0,1]$ on $\R^2\backslash B_1(0)$ and 
on $\R^2\backslash [-1,1]\times \R$ respectively.
\begin{theorem}\label{thm2}
Assume that there is a sequence
$(\bx_n)_{n\in\N}$ with $|\bx_n|\to\infty$ as $n\to\infty$ 
and constants $k\in \N$, $\varepsilon\in(0,1)$
such that $|\nabla V|, |\nabla B|$ vary with rate $0$
on $(\bx_n)_{n \in \N}$, as well as
\begin{eqnarray}
\label{con0a}
 &&V(\bx_n)\longrightarrow - \infty,\\
\label{con1a}
&& \frac{|\nabla B(\bx_n)|^2}{|B(\bx_n)|^{1-\varepsilon}},
\,\frac{|\nabla V(\bx_n)|^2}{|V(\bx_n)|^{1-\varepsilon}}
\longrightarrow 0,
\\
\label{con2a}
 &&V(\bx_n)+2k|B(\bx_n)|\longrightarrow 0
\end{eqnarray}
as $n \to\infty$. Then $0\in \sigma_{\rm ess}(H)$.
\end{theorem}
Let us now consider the case $V \gg B$ at infinity. The next 
two theorems state that the accumulation of eigenvalues intensifies, 
creating more points in the essential spectrum and closing spectral gaps.
\begin{theorem}\label{thm3}
Assume that there is a continuous path $\gamma: \R^+ \to \R^2$,
with $|\gamma(t)| \to \infty$ as $t \to \infty$, and constants
$\epsilon > 0$, $\nu \in [0,1]$ such that
$|\nabla V|, |\nabla B|$ vary with rate $\nu$ on 
$\mathrm{Im}(\gamma)$, as well as
\begin{eqnarray}
\label{con3.1}
 && \frac{V(\gamma(t))}{2|B(\gamma(t))|} \longrightarrow -\infty \\
\label{con3.2}
&& \left(
\frac{|\nabla B(\gamma(t))|}{|B(\gamma(t))|} +
\frac{|\nabla V(\gamma(t))|}{|V(\gamma(t))|} \right)
\left(\frac{|V(\gamma(t))|^3}{B^2(\gamma(t))}\right)^{\frac{1+\epsilon}{2}}
\longrightarrow 0,
\\[0.1cm]
\label{con3.3}
 &&\frac{1}{|\gamma(t)|^{2\nu}} 
\left(\frac{|V(\gamma(t))|}{B^2(\gamma(t))}\right)^{1+\epsilon}
\longrightarrow 0
\end{eqnarray}
as $t \to \infty$. In addition, suppose that 
for all $t \in (0,\infty)$ the inequality
  \begin{eqnarray}
   \label{con3.4}
&& B_0\le|B(\gamma(t))|\le 
\alpha \exp \left(\kappa\left| \frac{V(\gamma(t))}{B(\gamma(t))}\right|\right)
\end{eqnarray}
holds with constants $\alpha, \kappa, B_0 >0$. Then $\sigma_{\rm ess}(H) = \R$.
\end{theorem}
%%%%%%%%%%%%%THM4%%%%%%%%%%%%%%%%%%%%%%%%%%%%%%%%%%%%%%%%%%%%%%%%%%%%%%%%%%%%
For our main application, potentials of power-like growth 
(see discussion after the next theorem), condition \eqref{con3.3} 
imposes unsatisfying restrictions on the growth rate of $V/B$.
At least in the case of a constant magnetic field they can be 
weakened. 
\begin{theorem}\label{thm4}
Let $B = B_0 > 0$ and $V \in C^2(\R^2,\R)$. Assume that there is a
continuous path $\gamma: \R^+ \to \R^2$, with 
$|\gamma(t)| \to \infty$ as $t \to \infty$, and constants
$\epsilon >0 $, $\nu \in [0,1]$ such that
the matrix norm of the Hessian matrix 
$\|\mathrm{Hess}(V)\|_2 : \R^2 \to \R$
varies with rate $\nu$ on $\mathrm{Im}(\gamma)$, as well as
\begin{eqnarray}
\label{con4.1}
&& V(\gamma(t))\longrightarrow -\infty , \\[1.3 mm]
\label{con4.2}
&&\|\mathrm{Hess}(V)\|_2 (\gamma(t)) |V(\gamma(t))|^{1+\epsilon}
\longrightarrow 0,
\\
\label{con4.3}
 &&\frac{1}{|\gamma(t)|^{2\nu}} \,
|V(\gamma(t))|^{1+\epsilon}
\longrightarrow 0
\end{eqnarray}
as $t \to \infty$. In addition, let 
\begin{eqnarray}
\label{con4.4}
&&\limsup_{t \to \infty} \frac{|\nabla V(\gamma(t))|^2}{|V(\gamma(t))|}
< (2B_0)^2.
\end{eqnarray}
Then $\sigma_{\rm ess}(H) = \R$.
\end{theorem}
\begin{remark}
Note that a well-known, basic example for this last theorem
is the case of a constant electric field $ {\mathcal E}_0$ in 
$x_1$-direction with the corresponding 
potential $V(\bx) = {\mathcal E}_0 x_1$.  
\end{remark}
\begin{remark}
Results similar to that of Theorems \ref{thm1}$-$\ref{thm4} can be 
obtained for the magnetic Schr\"odinger operator with scalar 
potentials $V$ by using the same techniques as in the proofs of 
Theorems \ref{thm1}$-$\ref{thm4}.
\end{remark}
%%%%%%%%%%%%%%%%%%%%%%%%%%%%%%%%%%%%%%%%%%%%%%%%%%%%%%%%%%%%%%%%%%%%%%%%%%%%%%%
Finally, we want to discuss some consequences of our results,
in particular with respect to spherically symmetric fields $B$ and 
potentials $V$, i.e. $B(\bx) = b(|\bx|)$, $V(\bx) = v(|\bx|)$
for $\bx \in \R$. Using the rotational gauge
\begin{equation*}\label{goischt}
{\bf A}({\bf x}) := \frac{A(r)}{r} 
\begin{pmatrix} -x_2 \\ x_1 \end{pmatrix}\,, \quad
A(r) =  \frac{1}{r} \int_0^r b(s) s \mbox{d}s, 
\end{equation*}
with $r=|\bx|$, we decompose $H$ in a direct sum of
operators on the half-line. More explicitly, there is
a unitary map
$$U: L^2(\R^2,\C^2) \to \bigoplus_{j\in \Z}L^2(\R^+,\C^2; \mbox{d}r)$$
such that 
$ UHU^* =\bigoplus_{j\in \Z} h_j$,
with
\begin{align*}
h_j :=  
\begin{pmatrix}
  -\partial_r^2 + \frac{j^2-1/4}{r^2} & 0 \\
    0    & -\partial_r^2 + \frac{(j+1)^2-1/4}{r^2} 
\end{pmatrix}
 + A^2(r) - \frac{m_j}{r}A(r) + \sigma_3 A'(r)+ v(r) 
\end{align*}
on $L^2(\R^+,\C^2; \mbox{d}r)$, where $m_j = j +\tfrac{1}{2}$ 
(see e.g. \cite{Thaller}).  It is easy  to verify that if
\begin{eqnarray}
\label{condRot.1} 
&& \liminf_{r \to \infty} b(r) > 0, \\
\label{condRot.2} 
&&  A' (r) /A^2(r) \longrightarrow 0 \quad  \mbox{as} \ r \to \infty ,\\[1.5mm]
\label{condRot.3} 
&& \limsup_{r \to \infty} |v(r)|/A^2(r) < 1,
\end{eqnarray}
then $h_j$ has purely discrete spectrum for every $j \in \Z$.
As a consequence, one can use the relations
\begin{align*}
\sigma_{\#}(H)=
\overline{\bigcup_{j\in \Z} \sigma _{\#}{(h_j)}} 
 \,, \quad \quad \quad
\# \in \{ \mathrm{ac},\ \mathrm{sc}, \ \mathrm{pp} \}
\end{align*}
to conclude that $\sigma (H)=\sigma_{\mathrm{pp}}(H)$,
$\sigma_{\mathrm{ac}}(H)= \sigma_{\mathrm{sc}}(H) = \emptyset$
if \eqref{condRot.1}$-$\eqref{condRot.3} are satisfied. 
To get more information on $\sigma (H)=\sigma_{\mathrm{pp}}(H)$, 
we employ Theorems \ref{thm1}$-$\ref{thm4} and obtain:
\begin{corollary}
Let $b(r) = b_0 r^s$, $v(r) =v_0r^t$ with $v_0 < 0 < b_0$
and exponents $0\le s$, $0\le t \le 2$. Then
\begin{itemize}
\item[a)] $\sigma (H)$ is purely discrete if  $0< t < s$ or 
$0< t = s$ and $|v_0| < 2b_0$,
\item[b)] $0 \in \sigma_{\rm ess} (H)$ if $0< t=s$ and $|v_0| = 2kB_0$
for some $k \in \N$,
\item[c)] $\sigma (H) = \R$ is dense pure point 
if $3s<3t < 2(s+1)$,
\item[d)] $\sigma (H) = \R$ is dense pure point 
if $s=0$ and $0< t < 1$.
\end{itemize}
\end{corollary}
The origins of the strong restrictions on $s,t$ in c), d)
can easily be tracked back to conditions \eqref{con3.2},
\eqref{con3.3} of Theorem 
\ref{thm3} and \eqref{con4.2} of Theorem \ref{thm4}. Unfortunately,
even in the case of a constant magnetic field ($s = 0$) we cannot 
cover the full range of potentials ($0 < t \le 2$) for which one
might expect $\sigma (H) = \R$. \\ 

%%%%%%%%%%%%%%%%%%%%%%%%%%%%%%%%%%%%%%%%%%%%%%%%%%%%%%%%%%%%%%%%%%%%%%%%%%%%%%%
% As a second application we remark consequences for 
% problems with symmetry in $x_2$-direction, i.e. $V(\bx) = V(x_1)$,
% $B(\bx) = B(x_1)$. By choosing the Landau gauge 
% ${\bf A} (\bx) = \int_0^{x_1}B(s) \rm{d}s\  \hat {\bf e}_2$ we can decompose $H$
% as a direct integral
% $$H \cong \int_\R^\oplus h(\xi) \mathrm{d} \xi$$
% on $L^2(\R_\xi, L^2(\R,\C^2))$, with
% \begin{align*}
% h(\xi) = -\partial_1^2 + (\xi - A_2 )^2 + \sigma_3 B+ V.
% \end{align*}
%%%%%%%%%%%%%%%%%%%%%%%%%%%%%%%%%%%%%%%%%%%%%%%%%%%%%%%%%%%%%%%%%%%%%%%%%%%%%%%
%%%%%%%%%%%%%%%%%%%%%%%%%%%%%%%%%%%%%%%%%%%%%%%%%%%%%%%%%%%%%%%%%%%%%%%%%%%%%%%
%%%%%%%%%%%%%%%%%%%%%%%%%%%%%%%%%%%%%%%%%%%%%%%%%%%%%%%%%%%%%%%%%%%%%%%%%%%%%%%
%%%%%%%%%%%%%%%%%%%%%%%%%%%%%%%%%%%%%%%%%%%%%%%%%%%%%%%%%%%%%%%%%%%%%%%%%%%%%%%
%%%%%%%%%%%%%%%%%%%%%%%%%%%%%%%%%%%%%%%%%%%%%%%%%%%%%%%%%%%%%%%%%%%%%%%%%%%%%%%
\section{Proof of Theorem \ref{thm1}}\label{prooft1}
Note that the assumptions imply that either $B(\bx) \to \infty$ 
or $B(\bx) \to -\infty$. It suffices to consider the case $B(\bx)\to \infty$ 
as $|\bx|\to \infty$ since otherwise we only have to interchange the 
roles of $d$ and $d^*$ in the proof. By modifying $B$ and $V$ on a 
compact set and comparing the corresponding resolvents, we may 
assume that $B$ and $V$ satisfy
\begin{align}
\label{c1} &V(\bx) \le  - 1/\delta,\\
\label{c2} &|\nabla V(\bx)| \le |\delta V(\bx)|, \\
\label{c3} &|V(\bx)|\le 2(1-\eta)B(\bx),
\end{align} 
where $\delta \in (0,\tfrac{1}{4})$ is fixed, but can be chosen arbitrarily
small, and $\eta \in (0,1)$ is a fixed ($\delta$-independent) constant
(c.f. \cite{MehringStock} Appendix B).

Using the commutator relation \eqref{eq:1}, we see that 
\begin{align}\label{commuteddstar}
dd^* \ge 2B \ge (1- \eta)^{-1} |V| \ge (1-\eta)^{-1}\delta^{-1}
\end{align}
on $C^\infty_0(\R^2,\C)$ and therefore on $\mathcal{D}(dd^*)$.
Since $dd^*$ and $d^*d$ are isospectral away from $0$, we
obtain a spectral gap $(0,\beta) \subset \varrho(\HA)$,
with $\beta = (1-\eta)^{-1}\delta^{-1}$. Thus,
$0$ can be regarded as an isolated point of the spectrum,
which is used in the following commutator estimates.
%%%%%%%%%%%%%%%%%%%%%%%%%%%%%%%%%%%%%%%%%%%%%%%%%%%%%%%%%%%%%%%%%%%%%%%%
\begin{lemma}\label{lemmakey}
Let $V\in C^1(\R^2,\R)$, $B\in C(\R^2,\R)$ and 
${\bf A}\in C^1(\R^2,\R^2)$ with
  $B=\curla$. Assume further that the
  conditions \eqref{c1}$-$\eqref{c3} are fulfilled for
  $\delta\in(0,\tfrac{1}{4})$ and $\eta\in(0,1)$. Then:
  \begin{itemize}
  \item[a)] The operators $\left[ P_0^\perp,V^{-1} \right] V$,
    $V\left[ P_0^\perp,V^{-1}\right]$ are well-defined on  $\core$ and  
  extend to bounded operators on $\hilbert$ with 
    \begin{align*}
      \left \lVert V\left[P_0^\perp,V^{-1}\right]\right \rVert , 
      \left \lVert \left[P_0^\perp,V^{-1}\right] V\right \rVert
       \le 4 \delta^\frac{3}{2}.
    \end{align*}
The same holds true if we replace $P_0^\perp$ above by $P_0$. 
\item[b)] $P_0\mathcal{D}(V), P_0^\perp \mathcal{D}(V) \subset
  \mathcal{D}(V) $.
  \end{itemize}
\end{lemma}
%%%%%%%%%%%%%%%%%%%%%%%%%%%%%%%%%%%%%%%%%%%%%%%%%%%%%%%%%%%%%%%%%%%%%%%%%%%%
\begin{lemma}\label{lemmakey2} 
  Let $V\in C^1(\R^2,\R)$, $B\in C(\R^2,\R)$ and ${\bf A}\in
  C^1(\R^2,\R^2)$ with $B=\curla$. Assume
  further that the conditions \eqref{c1}$-$\eqref{c3} are fulfilled for
  $\delta\in(0,\tfrac{1}{4})$ and $\eta\in (0,1)$. Then 
$\left[{\rm sgn}(\DA)P_0^\perp,V^{-1} \right]$ maps $\hilbert$ into
 $\mathcal{D}(V)$ and 
\begin{equation}\label{rrr}
\left \lVert V\left[{\rm sgn}(\DA)P_0^\perp,V^{-1} \right] \right
\rVert \leq  4 \delta^\frac{3}{2}.
\end{equation}
\end{lemma}
The proofs of these commutator estimates can be found in
\cite{MehringStock}. Since $\DA$ is a first-order operator,
it is much more convenient to commute $V$ with functions 
of $\DA$ instead of with functions of $\HA$.
%%%%%%%%%%%%%%%%%%%%%%%%%%%%%%%%%%%%%%%%%%%%%%%%%%%%%%%%%%%%%%%%%%%%%%%%%%%%%
For proving Theorem \ref{thm1}, it suffices to find a constant 
$c>0$ such that 
\begin{align}
  \label{eq:2}
  \| H \varphi \|\ge c\|V\varphi\|,\quad \varphi\in \core
\end{align}
holds (see Lemma \ref{Gertrud} in the appendix). 
%%%%%%%%%%%%%%%%%%%%%%%%%%%%%%%%%%%%%%%%%%%%%%%%%%%%%%%%%%%%%%%%%%%%%%%%%%%%%%%
%%%%%%%%%%%%%%%%%%%%%%%%%%%%%%%%%%%%%%%%%%%%%%%%%%%%%%%%%%%%%%%%%%%%%%%%%%%%%%%
%%%%%%%%%%%%%%%%%%%%%%%%%%%%%%%%%%%%%%%%%%%%%%%%%%%%%%%%%%%%%%%%%%%%%%%%%%%%%%%
%%%%%%%%%%%%%%%%%%%%%%%%%%%%%%%%%%%%%%%%%%%%%%%%%%%%%%%%%%%%%%%%%%%%%%%%%%%%%%%
%%%%%%%%%%%%%%%%%%%%%%%%%%%%%%%%%%%%%%%%%%%%%%%%%%%%%%%%%%%%%%%%%%%%%%%%%%%%%%%
\begin{proof}[Proof of Theorem \ref{thm1}]
Let $\varphi\in\core$. By Lemma \ref{lemmakey}
we can split $\|H\varphi\|$ as 
\begin{equation*}
\begin{split}
  \left \lVert (\HA+V)\varphi \right \rVert ^2 
&= \left \lVert (\HA+V)(P_0+P_0^\perp)\varphi \right \rVert^2\\
&= \left \lVert \big(VP_0+(\HA+V)P_0^\perp \big)\varphi \right \rVert^2\\
&= \left \lVert(\HA+V)P_0^\perp\varphi \right \rVert^2+
2{\rm Re}
\sps{(\HA+V)P_0^\perp\varphi}{VP_0\varphi}+
\left \lVert VP_0\varphi \right \rVert^2\\
&=\left \|(\HA+V)P_0^\perp\varphi \right \|^2-\delta \left \lVert VP_0^\perp\varphi \right \rVert^2\\
&\hspace{1.37cm}+
2{\rm Re}\sps{VP_0\varphi}{\HA
  P_0^\perp\varphi}+\|V\varphi\|^2-(1-\delta) \left \lVert VP_0^\perp\varphi \right \rVert^2.
\end{split}
\end{equation*}
For the cross-term, condition \eqref{c2} yields
\begin{equation}
  \label{eq:17}
  \begin{split}
|\sps{VP_0\varphi}{\HA
  P_0^\perp\varphi}|&=|\sps{\DA VP_0\varphi}{\DA
  P_0^\perp\varphi}|\\
&= | \sps{ \left[\DA,V \right]V^{-1} V P_0\varphi}{\DA P_0^\perp\varphi}|\\
&\le \tfrac{1}{2}\delta^{-\frac{1}{2}} 
    \| (-\ri{\boldsymbol\sigma} \nabla V) V^{-1} V P_0\varphi\|^2
    + \tfrac{1}{2} \delta^{\frac{1}{2}}
     \| \DA P_0^\perp\varphi \|^2\\
&\le \tfrac{1}{2}\delta^{\frac{3}{2}} 
    \| V P_0\varphi\|^2
    + \tfrac{1}{2} \delta^{\frac{1}{2}}
       \| \HA P_0^\perp\varphi \| \|\varphi\|\\
&\le \tfrac{1}{4}\delta^{\frac{3}{2}} 
    \| \HA P_0^\perp\varphi \|^2
    + \tfrac{1}{4} \delta^{\frac{3}{2}}\big( \|V \varphi\|^2+2\|V P_0\varphi\|^2\big).
   \end{split}
\end{equation}
Applying Lemma \ref{lemmakey} a) results in
\begin{align*}
  \big \lVert VP_0^\perp\varphi \big \rVert , 
  \big \lVert VP_0\varphi \big \rVert      \le  
   \big(1+ 4\delta^\frac{3}{2} \big)\|V\varphi\|,
\end{align*}
and therefore
\begin{align}\label{eq:19}
\begin{split}  
\|V\varphi\|^2-(1-\delta)\left \| VP_0^\perp\varphi \right\|^2 
    - \tfrac{1}{4} \delta^{\frac{3}{2}} \|V \varphi\|^2
   &- \tfrac{1}{2} \delta^{\frac{3}{2}}  \|V P_0 \varphi\|^2 \\
   &\ge \big(\delta- 14\delta^\frac{3}{2}\big) \|V\varphi\|^2.
\end{split}
\end{align}
Because
\begin{equation*}
\begin{split}
\left\|(\HA+V)P_0^\perp\varphi\right\|^2 &- 
\delta^{\frac{3}{2}}  \| \HA P_0^\perp\varphi \|^2 -
\delta\|V P_0 \varphi\|^2 \\
& \ge (1-\varepsilon-\delta^{\frac{3}{2}}) 
\|\HA P_0^\perp\varphi\|^2+(1-\varepsilon^{-1}-\delta)
\|V P_0^\perp\varphi\|^2 
\end{split}
\end{equation*}
for any $\varepsilon\in (0,1)$, it suffices to show,
in view of \eqref{eq:17} and \eqref{eq:19},
that
\begin{align}
  \label{eq:20}
  \sps{\HA P_0^\perp\varphi}{\HA
    P_0^\perp\varphi}+\tfrac{1-\varepsilon^{-1}-\delta}
    {1-\epsilon-\delta^\frac{3}{2}}
     \sps{V P_0^\perp\varphi}{V P_0^\perp\varphi}\ge 0
\end{align}
for $\delta>0$ small enough and some $\varepsilon\in (0,1)$. 
We choose $\epsilon = 1 - \delta^\frac{1}{2}$, then
\begin{align*}
-\frac{1-\epsilon^{-1}-\delta}{1-\epsilon-\delta^\frac{3}{2}} = 
 \frac{1}{1-\delta}\left(\frac{1}{1-\delta^\frac{1}{2}} +\delta^\frac{1}{2}\right) 
  =: c_\delta >0.
\end{align*}
Since $dd^* \ge 2B$ and therefore ${\rm ker}(d^*) = \{0\}$, we have
\begin{align}
  \label{eq:21}
  P_0^\perp =
\begin{pmatrix}
\pi^\perp&0\\
0& \pi_*^\perp
\end{pmatrix} =
\begin{pmatrix}
\pi^\perp&0\\
0& \mathbbm{1}
\end{pmatrix}\!.
\end{align}
Setting
$\varphi=(\varphi_1,\varphi_2)^{\rm T}$, one can rewrite \eqref{eq:20} as
\begin{align*}
  \left \lVert \HA P_0^\perp\varphi \right \rVert^2- 
c_{\delta}\left \lVert V P_0^\perp\varphi \right \rVert^2=
\|d^*d\pi^\perp\varphi_1\|^2- c_{\delta}\|V\pi^\perp\varphi_1\|^2
+\|dd^*\varphi_2\|^2- c_{\delta}\|V\varphi_2\|^2.
\end{align*}
%%%%%%%%%%%%%%%%%%%%%%%%%%%%%%%%%%%%%%%%%%%%%%%%%%%%%%%%%%%%%%%%%%%%%%%%%%%%
By using the isometries $s,s^*$ given in \eqref{super0}, 
relation \eqref{super} and estimate \eqref{commuteddstar},
one obtains
\begin{align}\nonumber
  \|dd^*\varphi_2\|^2-
  c_{\delta}\|V\varphi_2\|^2
&=\sps{d^*\varphi_2}{d^*dd^*\varphi_2} -
  c_{\delta}\Sps{\sqrt{-V}\varphi_2}{|V|\sqrt{-V}\varphi_2}\\
\nonumber
&=\sps{sd^*\varphi_2}{dd^*sd^*\varphi_2} -
  c_{\delta}\Sps{\sqrt{-V}\varphi_2}{|V|\sqrt{-V}\varphi_2}\\
\nonumber
&\ge\sps{sd^*\varphi_2}{2Bsd^*\varphi_2} -
  c_{\delta}\Sps{d^*\sqrt{-V}\varphi_2}{d^*\sqrt{-V}\varphi_2}\\
\nonumber
&\ge\sps{sd^*\varphi_2}{2Bsd^*\varphi_2} -
  c_{\delta}\Sps{\sqrt{-V}d^*\varphi_2}{\sqrt{-V}d^*\varphi_2}\\
\nonumber & \hspace{3.15cm} 
- c_{\delta}\Sps{\big[d^*,\sqrt{-V}\,\big]\varphi_2}{\sqrt{-V}d^*\varphi_2}\\
\nonumber & \hspace{3.15cm} 
- c_{\delta}\Sps{\sqrt{-V}d^*\varphi_2}{\big[d^*,\sqrt{-V}\,\big]\varphi_2}\\
\nonumber & \hspace{3.15cm} 
- c_{\delta}\Sps{\big[d^*,\sqrt{-V}\,\big]\varphi_2}{\big[d^*,\sqrt{-V}\,\big]\varphi_2}\\
\nonumber
&\ge \big\|\sqrt{2B}sd^*\varphi_2 \big\|^2 -
   c_{\delta}\left(\big\|\big[d^*,\sqrt{-V}\,\big]\varphi_2 \big\|
                             + \big\|\sqrt{-V} d^*\varphi_2 \big\| \right)^2 \\
\nonumber
&\ge \big\|\sqrt{2B}sd^*\varphi_2 \big\|^2 -
       c_{\delta}\left(\delta \big\|sd^*\varphi_2 \big\|
     + \big(1+4\delta^\frac{3}{2}\big)\big\|\sqrt{-V} sd^*\varphi_2 \big\|  \right)^2 \\
\begin{split}\label{ddstar}
&\ge \big[ 1-c_{\delta}(1-\eta)\big(1+15\delta^\frac{3}{2}\big) \big] 
   \big\|\sqrt{2B}sd^*\varphi_2 \big\|^2,
\end{split}
\end{align}
where we applied the bound
$\big \lVert \sqrt{-V} \big[s\pi^\perp,\sqrt{-V}^{-1}\big] \big \rVert
\le 4\delta^2$. 
For the latter we write
\begin{align*}
  \sqrt{-V} \left[{\rm sgn}(\DA)P_0^\perp,\sqrt{-V}^{-1}\right] = 
\left(\begin{array}{cc}
0&\sqrt{-V}\left[s^*,\sqrt{-V}^{-1}\right]\\
\sqrt{-V}\left[s\pi^\perp, \sqrt{-V}^{-1}\right]&0
  \end{array}
\right)
\end{align*}
and therefore, by
Lemma \ref{lemmakey2} with $\sqrt{-V}$ instead of $V$, we get
\begin{align*}
\begin{split}
\label{eq:75} \left \lVert \sqrt{-V} \left[s\pi^\perp,\sqrt{-V}^{-1}\right] \right \rVert
&\le
    \left \lVert \sqrt{-V} \left[{\rm sgn}(\DA)P_0^\perp,\sqrt{-V}^{-1}\right] \right
    \rVert\le 4\delta^\frac{3}{2}.
\end{split}
\end{align*}
\\
Similarly, we obtain a lower bound for
$ \|d^*d\pi^\perp\varphi_1\|^2-c_{\delta}\|V\pi^\perp\varphi_1\|^2$ 
by using again the upper relation of Equation
\eqref{super}. More precisely,
\begin{equation*}\label{dddd}
\begin{split}
  \|d^*d\pi^\perp\varphi_1\|^2- c_{\delta}\|V\pi^\perp\varphi_1\|^2&=
  \|d^*ds^*s\pi^\perp\varphi_1\|^2- c_{\delta}\|Vs^*s\pi^\perp\varphi_1\|^2\\
&= \|dd^*s\pi^\perp\varphi_1\|^2- 
c_{\delta}\|Vs^*V^{-1} Vs\pi^\perp\varphi_1\|^2\\
&=\|dd^*s\pi^\perp\varphi_1\|^2- 
c_{\delta}\big\|\big(s^*+V\left[s^*,V^{-1}\right]\big)Vs\pi^\perp\varphi_1\big\|^2,
\end{split}
\end{equation*}
where 
$V\left[s^*,V^{-1}\right]$
is one of the components of the operator
\begin{align*}
  V \left[{\rm sgn}(\DA)P_0^\perp,V^{-1}\right]= \left(\begin{array}{cc}
0&V\left[s^*,V^{-1}\right]\\
V\left[s\pi^\perp, V^{-1}\right]&0
  \end{array}
\right)\!,
\end{align*}
so Lemma \ref{lemmakey2} yields
$\left \lVert V \left[s^*,V^{-1}\right] \right \rVert \le 4\delta^\frac{3}{2}$.
Thus, 
\begin{align*}
  \|d^*d\pi^\perp\varphi_1\|^2- c_{\delta}\|V\pi^\perp \varphi_1\|^2&
\ge \|dd^*s\pi^\perp\varphi_1\|^2 -
c_{\delta} (1+ 10 \delta^\frac{3}{2}) \big\|Vs\pi^\perp\varphi_1\big\|^2.
\end{align*}
We note that 
$s\pi^\perp\varphi_1 \subset \mathcal{D} (dd^*) 
\subset \mathcal{D}(V)$, hence we can use \eqref{ddstar}
(by approximating $s\pi^\perp\varphi_1$ through 
$C_0^\infty$-functions in the graph norm of $dd^*$)
to conclude that
\begin{align*}
  \|d^*d\pi^\perp\varphi_1\|^2 &- c_{\delta}\|V\pi^\perp \varphi_1\|^2 \\
& \ge \big[ 1-c_{\delta}(1-\eta)\big(1+10\delta^\frac{3}{2}\big) 
\big(1+15\delta^\frac{3}{2}\big)  \big] 
\big\|\sqrt{2B}sd^*s\pi^\perp \varphi_1 \big\|^2.
\end{align*}
Combining this inequality with \eqref{ddstar} leads to 
\begin{align*}
\left \lVert \HA P_0^\perp\varphi \right \rVert^2 & - 
c_{\delta}\left \lVert V P_0^\perp\varphi \right \rVert^2 \\
&\ge
 \big[ 1-c_{\delta}(1-\eta)(1+50\delta^\frac{3}{2})\big] 
\left(  \big\|\sqrt{2B}sd^*\varphi_2 \big\|^2 +
\big\|\sqrt{2B}d\pi^\perp\varphi_1 \big\|^2 \right)\!,
\end{align*}
where the r.h.s is non-negativ for $\delta $ small enough.
\end{proof} 
%%%%%%%%%%%%%%%%%%%%%%%%%%%%%%%%%%%%%%%%%%%%%%%%%%%%%%%%%%%%%%%%%%%%%%%%%%%%%%%%
%%%%%%%%%%%%%%%%%%%%%%%%%%%%%%%%%%%%%%%%%%%%%%%%%%%%%%%%%%%%%%%%%%%%%%%%%%%%%%%%
%%%%%%%%%%%%%%%%%%%%%%%%%%%%%%%%%%%%%%%%%%%%%%%%%%%%%%%%%%%%%%%%%%%%%%%%%%%%%%%%
%%%%%%%%%%%%%%%%%%%%%%%%%%%%%%%%%%%%%%%%%%%%%%%%%%%%%%%%%%%%%%%%%%%%%%%%%%%%%%%%
%%%%%%%%%%%%%%%%%%%%%%%%%%%%%%%%%%%%%%%%%%%%%%%%%%%%%%%%%%%%%%%%%%%%%%%%%%%%%%%%
\section{Proofs of Theorem \ref{thm2} and \ref{thm3}}\label{prooft2}
The basic strategy of the proofs is to represent $B$ and $V$ locally through
constant values $V_n:=V(\bx_n)$ and $B_n:=B(\bx_n)$ along a sequence 
$(\bx_n)_{n\in\N}\subset\R^2$. Since one also needs to compare
vector potentials associated to $B_n$ and $B$, we use the gauges
\begin{align*}
  &{\mathbf A}_n(\bx):=\int_0^1 B_n\wedge (\bx-\bx_n) s \rd
  s=\tfrac{1}{2} B_n\wedge (\bx-\bx_n) ,\\
  &\widetilde{{\mathbf A}}_n(\bx):= \int_0^1
  B(\bx_n+s(\bx-\bx_n))\wedge (\bx-\bx_n) s \rd s,
\end{align*}
where $a\wedge {\mathbf v}:=a(-v_2,v_1)$ for $a\in \R$ and 
${\mathbf v}=(v_1,v_2)\in \R^2$. The two given vector potentials satisfy
${\rm curl}\,{\mathbf A}_n =  {\rm curl}\,\widetilde{\mathbf A}_n =B $,
hence for every $n \in \N$ there exists a function 
$g_n\in C^2(\R^2,\R)$ such that $\nabla
  g_n={\mathbf A}-\widetilde{{\mathbf A}}_n$ .
In addition, for every vector potential ${\mathbf A}_n$, 
representing the constant magnetic fields $B_n$,
we obtain operators $d_n$ and $d_n^*$, $n \in \N$,
 defined as in \eqref{carolina}. 
For a sequence of natural numbers  $(\pn)_{n\in\N}$ 
we set
\begin{align*}
  \psi_n(\bx):=
\begin{pmatrix}
 (d_n^* )^\pn e^{-B_n|\bx-\bx_n|^2/4}\\
 0
\end{pmatrix}\!.
\end{align*}
Iterating commutator relation \eqref{eq:1}  for $d_n,d_n^*$ yields
\begin{align}
  \label{eq:25}
  d_n^* d_n  \big[(d_n^* )^{k_n} e^{-B_n|\bx-\bx_n|^2/4}\big]
  =2k_nB_n \big[(d_n^* )^{k_n} e^{-B_n|\bx-\bx_n|^2/4}\big], 
\quad n \in \N,
\end{align}
i.e. $ \psi_n$ is an eigenfunction of $H_{{\mathbf A}_n}$ with
the corresponding eigenvalue $2\pn B_n$. 
For the localization let $\chi \in C^\infty_0(\R^2,[0,1])$ 
be such that $\chi(\bx)=1$ for $|\bx|\le 1$
and $\chi(\bx)=0$ for $|\bx|\ge 2$. We set
\begin{align*}
  \chi_n(\bx):=\chi\left(\frac{\bx-\bx_n}{r_n}\right)\!,
  \end{align*} 
where the $r_n>0$ will be chosen in the proofs.
For the Weyl sequence we define the functions 
$\varphi_n$ through
\begin{align}
  \label{eq:26}
  \varphi_n(\bx):= e^{\ri g_n (\bx)}\chi_n(\bx) \psi_n(\bx),
  \quad \bx\in\R^2,
\end{align}
with $n \in \N$.  Bounds on the norm of $\varphi_n$ can be obtained
as in \cite{MehringStock}. They are given by:
\begin{lemma}\label{landau}
For all $n\in \N$ large enough we have
\begin{align}\label{eq:29a}
 & \|\varphi_n\|^2 \le \|\psi_n\|^2
= 2\pi \int_0^\infty (B_n r)^{2\pn} e^{-\frac{B_n}{2} r^2} \rd r 
= 2^{\pn+1} \pi  B_n^{\pn-1}\pn! , \\
\label{eq:29}
&\|\varphi_n\|^2\ge
 \|\psi_n\|^2\left(1-\frac{1}{\pn!}
\int_{\frac{1}{2}B_n  r_n^2}^\infty s^{\pn} e^{-s} \rd s\right)\!.
\end{align}
\end{lemma}
Now $H\varphi_n$ can be written as
\begin{equation}\label{weyl}
\begin{split}
 e^{-\ri g_n} (\HA+V)\varphi_n= \ &
(H_{\widetilde{\mathbf A}_n}\!+V) \chi_n \psi_n\\
= \ & (H_{{\mathbf A}_n}\!+V) \chi_n \psi_n+ 
2({\widetilde{\mathbf A}_n} -  {{\mathbf A}_n})
  ( -\ri \nabla -{\mathbf A}_n) \chi_n \psi_n \,+\\
& ({\widetilde{\mathbf A}_n} - {{\mathbf A}_n})^2\chi_n \psi_n 
   -\ri \nabla \cdot( {\widetilde{\mathbf A}_n} 
   - {{\mathbf A}_n})\chi_n\psi_n  \,+\\
& (B-B_n)\chi_n \psi_n  \,,
\end{split}
\end{equation}
with the localization error
\begin{align}
\begin{split} \label{localierror}
(H_{{\mathbf A}_n}\!+V) \chi_n \psi_n -\chi_n(H_{{\mathbf A}_n}\!
&+V)\psi_n \\
&=-(\Delta\, \chi_n)\psi_n + 
2(-\ri\nabla \chi_n)( -\ri \nabla -{\mathbf A}_n)\psi_n.
\end{split}
\end{align}
To prove Theorem \ref{thm2} and Theorem \ref{thm3}, 
we estimate each term of \eqref{weyl} separately. For
the proofs we use the notation 
$K_n := \{ x \in \R^2 \,|\, r_n \le |\bx-\bx_n| \le 2r_n\}$ 
with $n \in \N$.
%%%%%%%%%%%%%%%%%%%%%%%%%%%%%%%%%%%%%%%%%%%%%%%%%%%%%%%%%%%%%%%%%%%%%%%%%%%%%%%%%%%%%%%%%%%%%%%%%%%%%%%%%%%%%%%%%%
\begin{proof}[Proof of Theorem \ref{thm2}]
We set $k_n = k$ and choose the radii to be 
$r_n^{-4}=B_n^{(2-\epsilon)}$. Then, for any $p \ge 0$,
$$
\frac{(B_n)^p}{k_n!}
\int_{\frac{1}{2}B_n  r_n^2}^\infty s^{\pn} e^{-s} \rd s
=
\frac{(B_n)^p}{k!}
\int_{\frac{1}{2}B_n^{\epsilon/2} }^\infty s^{k} e^{-s} \rd s
\longrightarrow 0 \quad \mbox{as} \ n \to \infty.
$$
Further, we have 
 $\|\psi_n\|^2 \le 2\|\varphi_n\|^2$ for $n \in \N$ large enough.
For treating the terms on the r.h.s. of \eqref{weyl}, we estimate
\begin{align*}
\ \big\|({\widetilde{\mathbf A}_n}&  
- {{\mathbf A}_n})(-\ri\nabla  -{\mathbf A}_n) \chi_n \psi_n \big\|^2 
\\ \le
&  \ C_1 r_n^4 |\nabla B(\bx_n)|^2  
\big\|( -\ri \nabla -{\mathbf A}_n) \chi_n \psi_n \big\|^2 \\ \le
&  \ 2 C_1 r_n^4 |\nabla B(\bx_n)|^2 \bigg[(2k+1)B_n \| \psi_n \|^2  
 +  r_n^{-2} \|\nabla \chi \|_\infty^2 
  \int_{K_n} |\psi_n(\bx)|^2 \,\rd^2 x \bigg] \\ \le
&  \ 16kC_1 B_n\frac{|\nabla B(\bx_n)|^2}{B_n^{2-\epsilon}} \| \psi_n \|^2  
    +   4C_1\|\nabla \chi \|_\infty^2 \| \psi_n \|^2 B_n^{(2-\epsilon)/2}
      \frac{1}{k!}\int_{\frac{1}{2}B_n^{\epsilon/2}}^\infty s^ke^{-s}\rd s .
\end{align*}
In addition,
\begin{align*}
\|({\widetilde{\mathbf A}_n} -  {{\mathbf A}_n})^2 \chi_n \psi_n \|^2 
&\le C_2 r_n^4|\nabla B(\bx_n)|^4 \|\psi_n \|^2
\le C_2\frac{|\nabla B(\bx_n)|^4}{B_n^{2(1-\epsilon)}} 
B_n^{-\epsilon}\|\psi_n \|^2,\\
\| \mathrm{div\,} ( {\widetilde{\mathbf A}_n} - {{\mathbf A}_n}) \chi_n \psi_n  \|^2 
&= \| \mathrm{div\,}{\widetilde{\mathbf A}_n} \chi_n \psi_n  \|^2
\le C_3r_n^2 |\nabla B(\bx_n)|^2 \|\psi_n \|^2 \\
&\hspace{2.85cm}\le C_3 \frac{|\nabla B(\bx_n)|^2}{B_n^{1-\epsilon}}  
B_n^{-\epsilon/2} \|\psi_n \|^2,\\
\|(B-B_n)\chi_n \psi_n  \|^2 
&\le C_4r_n^2 |\nabla B(\bx_n)|^2 \|\psi_n \|^2
\le C_4\frac{|\nabla B(\bx_n)|^2}{B_n^{1-\epsilon}}  B_n^{-\epsilon/2}\|\psi_n \|^2.
\end{align*}
For the first term of the r.h.s of \eqref{weyl} we get, 
due to \eqref{localierror},
\begin{align*}
\|(H_{{\mathbf A}_n}\!+V) &\chi_n \psi_n \| \\
&\le \|\chi_n(H_{{\mathbf A}_n}\!+V) \psi_n \|
+ \|(\Delta\, \chi_n)\psi_n\|  + 
2\|(-\ri\nabla \chi_n)( -\ri\nabla -{\mathbf A}_n)\psi_n\|,
\end{align*}
with 
\begin{align*}
\|(\Delta\, \chi_n)\psi_n\|^2  
& \le  r_n^{-4} \|\Delta \chi \|_\infty^2 
\int_{K_n} |\psi_n(\bx)|^2 \,\rd^2 x \\
& \le 2\|\Delta \chi \|_\infty^2 \| \psi_n \|^2\frac{1}{k!}
    B_n^{2-\epsilon}\int_{\frac{1}{2}B_n^{\epsilon/2}}^\infty s^ke^{-s}\rd s, 
\end{align*}
and
\begin{align*}
\|(-\ri\nabla \chi_n)( -\ri\nabla -{\mathbf A}_n)\psi_n\|^2
&\le r_n^{-2} \|\nabla \chi \|_\infty^2 
\int_{K_n} \big|(-\ri\nabla -{\mathbf A}_n)\psi_n(\bx)\big|^2 \rd^2 x \\
&\le \|\nabla \chi \|_\infty^2 B_n^{(2-\epsilon)/2}(2k+1)B_n
\int_{K_n} \big|\psi_n(\bx)\big|^2 \rd^2 x \\
&\le \|\nabla \chi \|_\infty^2\| \psi_n \|^2 (2k+1) B_n^{2-\epsilon/2} 
\int_{\frac{1}{2}B_n^{\epsilon/2}}^\infty s^ke^{-s}\rd s. 
\end{align*}
Because of \eqref{eq:25} and since $|\nabla V|$ vary with rate $0$, 
we conclude by the mean value theorem that
\begin{align*}
\|\chi_n(H_{{\mathbf A}_n}\!+V) \psi_n \|^2 
&\le \|\chi_n(V+2kB_n) \psi_n \|^2 \\[0.2cm]
&\le C_5 |\nabla V(\bx_n)|^2 r_n^2 \|\chi_n \psi_n \|^2 + 
  (2kB_n + V_n)^2 \|\chi_n \psi_n \|^2 \\
&\le C_5 (4k)^{1-\epsilon} \frac{|\nabla V(\bx_n)|^2}{|V_n|^{1-\epsilon}}\|\varphi_n \|^2\
+ (2kB_n + V_n)^2 \|\varphi_n \|^2\,.
\end{align*}
Hence, by \eqref{weyl} and conditions \eqref{con0a}$-$\eqref{con2a},
we see that $ \|(\HA+V)\varphi_n\|/ \|\varphi_n\| \to 0 $ 
as $n \to \infty$. In addition, note that $r_n \to 0$ as $n \to \infty$, 
so we can assume that the $\varphi_n$ have mutually disjoint support,
i.e. $(\varphi_n)_{n \in \N}$ is a Weyl sequence for $0$. 
\end{proof}
%%%%%%%%%%%%%%%%%%%%%%%%%%%%%%%%%%%%%%%%%%%%%%%%%%%%%%%%%%%%%%%%%%%%%%%%%%%%%%%%
\begin{proof}[Proof of Theorem \ref{thm3}]
We first note that it suffices to proof  $0 \in \sigma_{\rm{ess}}(H)$ since for 
$E \in \R$ we consider $V_E := V-E$ instead of $V$, which also fulfills 
\eqref{con3.1}$-$\eqref{con3.4} along $\gamma$. Because
$\R^+ \ni t \mapsto V(\gamma(t))/B((\gamma(t))$ is continuous and
\eqref{con3.1} holds, we find points $(\bx_n)_{n \in \N} \subset \rm{Im}(\gamma)$,
with $|\bx_n| \to \infty$ as $n \to \infty$, such that $2nB(\bx_n)=-V(\bx_n)$. 
We choose $\pn = n$ and set 
 \begin{align}
  \label{eq:rad2} r_n := \sqrt{2n^{1+\epsilon}/B_n} .
 \end{align}
Note that $r_n/|\bx_n|^\nu \to 0$ as $n \to \infty$ by \eqref{con3.3}.
In particular, we might assume the $\varphi_n$'s to have mutually
disjoint support.
Further, for any $\lambda \ge 0$,  
$$\frac{e^{\lambda n}}{n!}\int_{n^{1+\epsilon}}^\infty s^n e^{-s} \rd s 
  \le  e^{\lambda n} \exp(n\ln(2n)-n^{1+\epsilon}/2 +n)\longrightarrow 0
 \quad \mbox{as} \ n \to \infty.$$
Hence, we can choose $N \in \N $ so large that 
$\|\varphi_n\|^2 \le \|\psi_n\|^2 \le 2 \|\varphi_n\|^2$ for $n \ge N$. Proceeding as
in the proof of Theorem \ref{thm2},  we obtain
\begin{align*}
\|({\widetilde{\mathbf A}_n} &-  {{\mathbf A}_n})(-\ri \nabla  -{\mathbf A}_n) 
  \chi_n \psi_n \|^2 \\ \le
& \ C_6 r_n^4 |\nabla B(\bx_n)|^2 (2n+1)B_n \bigg[\|\psi_n \|^2  
  + r_n^{-2} \|\nabla \chi \|_\infty^2 
   \int_{K_n} |\psi_n(\bx)|^2 \rd^2 x \bigg] \\ \le
& \ 16C_6n^{3+2\epsilon} B_n\frac{|\nabla B(\bx_n)|^2}{B_n^2} \| \psi_n \|^2  
 +   C_6\|\nabla \chi \|_\infty^2 \| \psi_n \|^2
    B_n\frac{1}{n!}\int_{n^{1+\epsilon}}^\infty s^ne^{-s}\rd s \\ \le
& \ \widetilde C_6\left(\frac{|V_n|^3}{B_n^2}\right)^{1+\epsilon}
       \frac{|\nabla B(\bx_n)|^2}{B_n^2} \| \psi_n \|^2  
  + \alpha C_6\|\nabla \chi \|_\infty^2 \| \psi_n \|^2
     \frac{e^{2\kappa n}}{n!}\int_{n^{1+\epsilon}}^\infty s^ne^{-s}\rd s.
\end{align*}
Using (as in the first inequality above) that $|\nabla B|$ vary with 
rate $\nu$, we conclude
\begin{align*}
\|({\widetilde{\mathbf A}_n} -  {{\mathbf A}_n})^2 \chi_n \psi_n \|^2 
&\le C_7 r_n^4|\nabla B(\bx_n)|^4 \|\psi_n \|^2
\le \frac{C_7}{4}\frac{|V_n|^6}{B_n^4}\frac{|\nabla B(\bx_n)|^4}{B_n^{4}}
\|\psi_n \|^2, \\
\|{\rm div}\, ( {\widetilde{\mathbf A}_n} - {{\mathbf A}_n}) 
\chi_n \psi_n  \|^2 
&\le C_8r_n^2 |\nabla B(\bx_n)|^2 \|\psi_n \|^2 
\le \frac{C_8}{2}\frac{|V_n|^3}{B_n^2}\frac{|\nabla B(\bx_n)|^2}{B_n^2}  \|\psi_n \|^2 ,
\end{align*}
\begin{align*}
\hspace{1.17cm}\|(B-B_n)\chi_n \psi_n  \|^2 
\le C_9r_n^2 |\nabla B(\bx_n)|^2 \|\psi_n \|^2
\le \frac{C_9}{2}\frac{|V_n|^3}{B_n^2}\frac{|\nabla B(\bx_n)|^2}{B_n^2}
\|\psi_n \|^2. \hspace{0.71cm}
\end{align*}
We estimate, using equality \eqref{localierror},
the first term on the r.h.s. of \eqref{weyl} by  
\begin{align*}
\|(H_{{\mathbf A}_n}\!+V) &\chi_n \psi_n \| \\ 
&\le \|(V-V_n)\chi_n \psi_n \| + \|(\Delta\, \chi_n)\psi_n\|  
+ 2\|(-\ri\nabla \chi_n)( -\ri \nabla - {\mathbf A}_n)\psi_n\|,
\end{align*}
with 
\begin{align*}
\|(\Delta\, \chi_n)\psi_n\|^2  
& \le \|\Delta \chi \|_\infty^2 \| \psi_n \|^2 \frac{B_n^2}{n^{2+2\epsilon}}
  \frac{1}{n!}\int_{n^{1+\epsilon}}^\infty s^ne^{-s}\rd s \\
& \le \alpha^2 \|\Delta \chi \|_\infty^2 \| \psi_n \|^2 
 \frac{e^{4\kappa n}}{n^{2+2\epsilon}}\frac{1}{n!}
\int_{n^{1+\epsilon}}^\infty s^n e^{-s}\rd s,
\end{align*}
and 
\begin{align*}
\|(-\ri\nabla \chi_n)( -\ri\nabla -{\mathbf A}_n)\psi_n\|^2
&\le \|\nabla \chi \|_\infty^2 \frac{B_n}{n^{1+\epsilon}}
\int_{K_n}  (2n+1)B_n|\psi_n(\bx)|^2 \rd^2 x \\
&\le \alpha\|\nabla \chi \|_\infty^2\| \psi_n \|^2 
 \frac{e^{2\kappa n}}{n^{1+\epsilon}}\frac{1}{n!}
\int_{n^{1+\epsilon}}^\infty s^ne^{-s}\rd s.
\end{align*}
Since $|\nabla V|$ vary with rate $\nu$, we have
\begin{align*}
\|\chi_n(H_{{\mathbf A}_n}\!+V) \psi_n \|^2 
&\le \|\chi_n(V-V_n) \psi_n \|^2 \\
&\le C_{10} |\nabla V(\bx_n)|^2 r_n^2 \|\psi_n \|^2 
\le \widetilde{C}_{10} \left(\frac{|V_n|^3}{B_n^2}\right)^{1+\epsilon}
 \frac{|\nabla V(\bx_n)|^2}{|V_n|^2}\|\psi_n \|^2.
\end{align*}
We see that $\|(H_{{\mathbf A}_n}\!+V) \varphi_n \|/ \|\varphi_n \| \to 0 $
as $n \to \infty$ and therefore, by \eqref{weyl} and the estimates
above, that $\|(\HA+V) \varphi_n \|/ \|\varphi_n \| \to 0 $  as 
$n \to \infty$.
\end{proof}
%%%%%%%%%%%%%%%%%%%%%%%%%%%%%%%%%%%%%%%%%%%%%%%%%%%%%%%%%%%%%%%%%%%%%%%%%%%%%%
%%%%%%%%%%%%%%%%%%%%%%%%%%%%%%%%%%%%%%%%%%%%%%%%%%%%%%%%%%%%%%%%%%%%%%%%%%%%%%
%%%%%%%%%%%%%%%%%%%%%%%%%%%%%%%%%%%%%%%%%%%%%%%%%%%%%%%%%%%%%%%%%%%%%%%%%%%%%%
%%%%%%%%%%%%%%%%%%%%%%%%%%%%%%%%%%%%%%%%%%%%%%%%%%%%%%%%%%%%%%%%%%%%%%%%%%%%%%
%%%%%%%%%%%%%%%%%%%%%%%%%%%%%%%%%%%%%%%%%%%%%%%%%%%%%%%%%%%%%%%%%%%%%%%%%%%%%%
\section{Proof of Theorem \ref{thm4}}\label{prooft4}
Throughout this section we consider the case of a constant 
magnetic field $B(\bx) = B_0$. In addition, we assume that 
${\bf A}$ is in the rotational gauge, i.e.
$${\bf A}(\bx) = \frac{B_0}{2}
\begin{pmatrix}
-x_2 \\ x_1
\end{pmatrix}\!.
$$
Note that $\HA$ is invariant under rotations. More precisely,
for a special orthogonal matrix ${\mathcal R} \in SO(2,\R)$ 
define the unitary map
\begin{align*}
U_{\mathcal R}  : \hilbert \to \hilbert,  \quad \quad \quad
\psi(\,\cdot\,) \mapsto \psi({\mathcal R}^{-1} \,\cdot\,) ,
\end{align*}
then $U_{\mathcal R}^{-1}\HA U_{\mathcal R} = \HA$ and therefore
\begin{align*}
U_{\mathcal R}^{-1} (\HA+V) U_{\mathcal R} = \HA + V_{\mathcal R} \quad \mbox{with} \
V_{\mathcal R}(\,\cdot\,) = V({\mathcal R} \,\cdot\,) .
\end{align*}
To construct a Weyl sequence, consider a second gauge 
$\widetilde{\bf A}(\bx) = B_0 x_1 \hat e_2 $, called the 
Landau gauge.  Then our Hamiltonian reads
\begin{align}\label{HwcE}
\begin{split}
H_{\widetilde{\bf A}} +V 
&= -\partial_1^2 + (-\ri \partial_2 -B_0x_1)^2 
- \sigma_3 B_0+V \\ 
&={\tilde d}^{\,*} {\tilde d} +B_0 - \sigma_3 B_0+V  ,
\end{split}
\end{align}
with
\begin{align*}
{\tilde d} =
-\ri \partial_1 + \ri (-\ri\partial_2 -B_0x_1) , \quad
{\tilde d}^{\,*} =
-\ri \partial_1 - \ri (-\ri\partial_2 -B_0x_1) .
\end{align*}
For electric fields of the form $V(\bx) = V_0 + {\mathcal E}_0 (x_1-\zeta)$, 
with constants $V_0, {\mathcal E}_0, \zeta \in \R$, we can write
\begin{align}
\begin{split}
H_{\widetilde{\bf A}} +V
&= -\partial_1^2 + (-\ri \partial_2 -B_0x_1)^2 - 
           \sigma_3 B_0+V_0 + {\mathcal E}_0 (x_1- \zeta) \\
& = -\partial_1^2 + 
           B_0^2\big(x_1 -\tfrac{1}{B_0}\big(-\ri \partial_2 - 
           \tfrac{ {\mathcal E}_0}{2B_0}\big)\big)^2 \\
&\hspace{1.11cm}
+ \tfrac{ {\mathcal E}_0}{B_0} \big( -\ri\partial_2
-\tfrac{ {\mathcal E}_0}{2B_0}\big) 
- {\mathcal E}_0 \zeta - \sigma_3 B_0+V_0 
+\big(\tfrac{ {\mathcal E}_0}{2B_0}\big)^2.
\end{split}
\end{align}
Performing a Fourier transform in $x_2$, we obtain the direct 
integral representation
\begin{align*}
H_{\widetilde{\bf A}} +V
 \cong \int_\R^\oplus h(\xi) \mathrm{d} \xi
\end{align*}
on $L^2(\R_\xi, L^2(\R,\C^2))$, with
\begin{align*}
 h(\xi) &=
-\partial_1^2 + 
B_0^2\big(x_1 -\tfrac{1}{B_0}\big( \xi -\tfrac{ {\mathcal E}_0}{2B_0}\big)\big)^2 + 
\tfrac{ {\mathcal E}_0}{B_0} \big( \xi-\tfrac{ {\mathcal E}_0}{2B_0}\big) 
 -  {\mathcal E}_0 \zeta - \sigma_3 B_0+V_0
 +\big(\tfrac{ {\mathcal E}_0}{2B_0}\big)^2 \\ 
&=
-\partial_1^2 + 
B_0^2\big(x_1 -\hat\zeta \big)^2 + 
 {\mathcal E}_0\hat \zeta - {\mathcal E}_0 \zeta - \sigma_3 B_0+V_0
 +\big(\tfrac{ {\mathcal E}_0}{2B_0}\big)^2 .
\end{align*}
Here we set $\hat \zeta = 
\tfrac{1}{B_0}\big( \xi -\tfrac{ {\mathcal E}_0}{2B_0}\big)$.
Note that  $h(\xi)$ is the Hamiltonian of a shifted 
harmonic oscillator. Thus, we define for $n \in \N_0$
$$\phi_n(x) = 
\frac{1}{\sqrt{2^nn!\sqrt{\pi} }}\, \vartheta_n(x)e^{-x^2/2}, 
\quad x \in \R,$$
where $\vartheta_n$ denotes the $n-$th Hermite polynomial. 
The normalized functions
\begin{align*}
{\widehat \psi}_{{\mathcal E}_0, n,\xi}(x_1) :=
\sqrt[4] {B_0}
\begin{pmatrix}
\phi_n\big(\sqrt{B_0}\big(x_1- \tfrac{1}{B_0}\big( \xi
-\tfrac{ {\mathcal E}_0}{2B_0}\big)\big)  \\
0
\end{pmatrix}
\end{align*}
fulfill the equation
\begin{align*}
h(\xi) {\widehat \psi}_{{\mathcal E}_0, n, \xi}=
\Big(2nB_0 + {\mathcal E}_0 \big( 
\tfrac{1}{B_0}\big( \xi-\tfrac{ {\mathcal E}_0}{2B_0}\big) -\zeta\big)
+ V_0+ \big(\tfrac{ {\mathcal E}_0}{2B_0}\big)^2 \Big)
{\widehat \psi}_{{\mathcal E}_0, n, \xi}\,.
\end{align*}
Hence,
\begin{align}\label{monika}
{\psi}_{{\mathcal E}_0, n,\xi}(x_1, x_2) :=
e^{\ri \xi x_2} {\widehat \psi}_{{\mathcal E}_0, n, \xi} (x_1) 
\end{align}
satisfies
\begin{align}\label{dominika}
\begin{split}
\big[H_{\widetilde{\bf A}} +V\big]{\psi}_{{\mathcal E}_0, n, \xi} = 
\Big(2nB_0 + {\mathcal E}_0 \big( 
\tfrac{1}{B_0}\big( \xi-\tfrac{ {\mathcal E}_0}{2B_0}\big) -\zeta\big) +
 V_0+ \big(\tfrac{ {\mathcal E}_0}{2B_0}\big)^2 \Big)
{ \psi}_{{\mathcal E}_0, n, \xi} 
\end{split}
\end{align}
for $\xi \in \R$ and $n \in \N_0$, seen as a differential equation.
In addition, we have
\begin{align}\label{annihilation}
{\tilde d} {\psi}_{{\mathcal E}_0, n,\xi}  & =
-\ri \sqrt{2nB_0} {\psi}_{{\mathcal E}_0, n-1,\xi}
+\ri\tfrac{ {\mathcal E}_0}{2B_0}{\psi}_{{\mathcal E}_0, n,\xi} \,,\\
\label{creation}
{\tilde d}^{\,*} {\psi}_{{\mathcal E}_0, n,\xi}   & =
\ri \sqrt{2(n+1)B_0} {\psi}_{{\mathcal E}_0, n+1,\xi}
-\ri\tfrac{ {\mathcal E}_0}{2B_0}{\psi}_{{\mathcal E}_0, n,\xi} \,.
\end{align}
%%%%%%%%%%%%%%%%%%%%%%%%%%%%%%%%%%%%%%%%%%%%%%%%%%%%%%%%%%%%%%%%%%%%%%%%%%%%%%%
\begin{proof}[Proof of Theorem \ref{thm4}]
As argumented in the proof of Theorem \ref{thm3}, it  
suffices to find a Weyl sequence for $E =0$. Because of
\eqref{con4.1} and \eqref{con4.4}, there exists a sequence
$\{ \by_n\}_{n \in \N} \subset \mbox{Im}(\gamma)$ such that
\begin{align}\label{crosscond}
V(\by_n) = -2nB_0 - \left(\tfrac{|\nabla V(\by_n)|}{2B_0}\right)^2.
\end{align}
Further, there are rotations ${\mathcal R}_n \in SO(2,\R)$ such that
$\nabla V_{{\mathcal R}_n}(\bx_n) 
= |\nabla V_{{\mathcal R}_n}(\bx_n)|\hat e_1$,  with
$\bx_n = {\mathcal R}_n^{-1}\by_n =(x_{n,1}, x_{n,2})^T$ 
for $n \in \N$. We set
\begin{align}\label{Vn}
\hspace{1.0cm}
&V_n := V(\by_n) = V_{{\mathcal R}_n} (\bx_n) ,\\
\label{En}
&{\mathcal E}_n := |\nabla V(\by_n)| = 
|\nabla V_{{\mathcal R}_n} (\bx_n)|  ,
\end{align}
\begin{align}\label{Xin}
\xi_n := B_0x_{n,1} + \tfrac{ {\mathcal E}_n}{2B_0} .
\hspace{0.71cm}
\end{align}
For the Weyl functions let $\chi \in C_0^\infty(\R, [0, 1])$ with
$\chi(x) =1$ for $|x| \le 1$ and $\chi(x) =0$ for $|x| \ge 2$. Define
\begin{align*}
\chi_{n,j}(x) := \chi \left( \tfrac{x -x_{n,j}}{r_n} \right)\!,
\end{align*} 
for $j = 1,2 $, and
\begin{align}\nonumber
\begin{split}
\varphi_n(\bx) &:= 
\chi_{n,1}(x_1)
\chi_{n,2}(x_2)
{\psi}_{{\mathcal E}_n, n,\xi_n}(x_1, x_2)
\\ & \ =
\chi\left( \tfrac{x_2 -x_{n,2}}{r_n}\right)
e^{-\ri \xi_n x_2} 
\chi\left( \tfrac{x_1 -x_{n,1}}{r_n}\right)
\begin{pmatrix}
\sqrt[4]{B_0}\phi_n\big(\sqrt{B_0}(x_1-x_{n,1})\big)
\\ 0
\end{pmatrix}\!,
\end{split}
\end{align}
where the localization radii $r_n$ are chosen to be 
$r_n := \sqrt{n^{1+\epsilon}/B_0}$.
Note that
\begin{align}\label{normthm4}
r_n \le 
2 r_n \int_{-\sqrt{n^{1+\epsilon}}}^{\sqrt{n^{1+\epsilon}}} |\phi_n(x)|^2 \rd x
\le \  \|\varphi_n\|^2 \ \le 
4 r_n \int_{-2\sqrt{n^{1+\epsilon}}}^{2\sqrt{n^{1+\epsilon}}} |\phi_n(x)|^2 \rd x 
\le 4 r_n
\end{align}
for $n \in \N$ large enough 
(see Lemma \ref{hermiteestimate} in the appendix).
By denoting $g(\bx) = \tfrac{B_0}{2}x_1x_2$ for $\bx \in \R^2$, we get, 
due to \eqref{HwcE}, \eqref{dominika}, \eqref{crosscond} and\eqref{Xin},
that
\begin{align}\label{applyphi}
\begin{split}
H U_{{\mathcal R}_n} e^{-\ri g} \varphi_n = &
\ U_{{\mathcal R}_n} e^{-\ri g} 
\big[H_{\widetilde{\bf A}} +V_{{\mathcal R}_n}  \big] \varphi_n 
\\ = &\
U_{{\mathcal R}_n} e^{-\ri g}\big(
{\tilde d}^{\,*}{\tilde d} \varphi_n -
\chi_{n,1}\chi_{n,2}{\tilde d}^{\,*}{\tilde d}
 {\psi}_{{\mathcal E}_n, n,\xi_n}\big) \ + \\  & \
U_{{\mathcal R}_n} e^{-\ri g} 
\big[V_{{\mathcal R}_n} - V_n  -
{\mathcal E}_n(x_1-x_{1,n})\big] \varphi_n .
\end{split}
\end{align}
The localization error results in  
\begin{align*}%\label{localerrorthm4}
\begin{split}
{\tilde d}^{\,*}{\tilde d} \varphi_n -
\chi_{n,1}\chi_{n,2}{\tilde d}^{\,*}{\tilde d}
{\psi}_{{\mathcal E}_n, n,\xi_n}  \ = \
&\big[-\ri\chi_{n,2} \partial_1\chi_{n,1} +
\chi_{n,1}\partial_2\chi_{n,2}\big]
{\tilde d}^{\,*} {\psi}_{{\mathcal E}_n, n,\xi_n} \ + \\
&\big[-\ri\chi_{n,2} \partial_1\chi_{n,1} -
\chi_{n,1}\partial_2\chi_{n,2}\big]
{\tilde d}{\psi}_{{\mathcal E}_n, n,\xi_n} \ + \\
&\big[-\chi_{n,2} \partial_1^2\chi_{n,1} -\chi_{n,1} \partial_2^2\chi_{n,2}\big]
{\psi}_{{\mathcal E}_n, n,\xi_n} ,
\end{split}
\end{align*}
with, using \eqref{annihilation} and \eqref{creation},
\begin{align*}
\big\|
\big[-\ri\chi_{n,2} \partial_1&\chi_{n,1} +
\chi_{n,1}\partial_2\chi_{n,2}\big]
{\tilde d}^{\,*}  {\psi}_{{\mathcal E}_n, n,\xi_n} \big \|  \\
&\le  \sqrt{2(n+1)B_0}\big\|
\big[-\ri\chi_{n,2} \partial_1\chi_{n,1} +
\chi_{n,1}\partial_2\chi_{n,2}\big] {\psi}_{{\mathcal E}_n, n+1,\xi_n}
\big \|  \\ & \hspace{1.52cm}+  
\tfrac{ {\mathcal E}_n}{2B_0}\big\|
\big[-\ri\chi_{n,2} \partial_1\chi_{n,1} +
\chi_{n,1}\partial_2\chi_{n,2}\big]{\psi}_{{\mathcal E}_n, n,\xi_n}
\big \| \\ &\le
2\sqrt{2(n+1)B_0} r_n^{-1}\|\chi'\|_\infty \sqrt{2r_n}
  \|\phi_{n+1}\| +
2\tfrac{ {\mathcal E}_n}{2B_0}
 r_n^{-1}\|\chi'\|_\infty \sqrt{2r_n} \|\phi_n\| \\
&\le 2\sqrt{2} \|\chi'\|_\infty 
\Big( B_0\sqrt{\tfrac{2n+2}{n^{1+\epsilon}}} +    
\tfrac{{\mathcal E}_n}{2}
\sqrt{\tfrac{B_0}{n^{1+\epsilon}}} \Big) \sqrt{r_n} \,,
\end{align*}
\begin{align*}
\big\|
\big[-\ri\chi_{n,2} \partial_1\chi_{n,1} -
\chi_{n,1}\partial_2&\chi_{n,2}\big]
{\tilde d}  {\psi}_{{\mathcal E}_n, n,\xi_n} \big \|   \\
&\le  \sqrt{2nB_0}\big\|
\big[-\ri\chi_{n,2} \partial_1\chi_{n,1} -
\chi_{n,1}\partial_2\chi_{n,2}\big] {\psi}_{{\mathcal E}_n, n-1,\xi_n}
\big \|  \hspace{0.27cm}
\\ & \hspace{0.63cm}+  
\tfrac{ {\mathcal E}_n}{2B_0}\big\|
\big[-\ri\chi_{n,2} \partial_1\chi_{n,1} -
\chi_{n,1}\partial_2\chi_{n,2}\big]{\psi}_{{\mathcal E}_n, n,\xi_n}
\big \| \\ &\le
2\sqrt{2} \|\chi'\|_\infty 
\Big( B_0\sqrt{\tfrac{2n}{n^{1+\epsilon}}} +    
\tfrac{{\mathcal E}_n}{2}
\sqrt{\tfrac{B_0}{n^{1+\epsilon}}} \Big) \sqrt{r_n} 
\end{align*}
and 
\begin{align*}
\big\|
\big[-\chi_{n,2} \partial_1^2\chi_{n,1} -\chi_{n,1} \partial_2^2\chi_{n,2}\big]
{\psi}_{{\mathcal E}_n, n,\xi_n}
\big \| \le
2 \sqrt{2}\|\chi''\|_\infty r_n^{-2} \sqrt{r_n} \,. 
\end{align*}
Thus, in view of condition \eqref{con4.4} and estimate
\eqref{normthm4}, we get
\begin{align}\label{normlocalierror} 
\big\|
{\tilde d}^{\,*}{\tilde d} \varphi_n -
\chi_{n,1}\chi_{n,2}{\tilde d}^{\,*}{\tilde d}
{\psi}_{{\mathcal E}_n, n,\xi_n} \big\| \big/
\|\varphi_n\| \longrightarrow 0
\quad \mathrm{as} \ n \to \infty .
\end{align}
%%%%%%%%%%%%%%%%%%%%%%%%%%%%%%%%%%%%%%%%%%%%%%%%%%%%%%%%%%%%%%%%%%%%%%%%%%

For estimating the remaining term on the r.h.s of \eqref{applyphi},
we expand $V_{{\mathcal R}_n}$ up to second order and obtain, 
by \eqref{Vn} and \eqref{En}, that
\begin{align*}
\big| \big[ V_{{\mathcal R}_n}(\bx) - V_n  -
{\mathcal E}_n(x_1-x_{1,n})\big] \varphi_n(\bx) \big|
& \le
\big\| \mathrm{Hess} (V_{{\mathcal R}_n}) \big\|_2
(\boldsymbol{\eta_{x,x_n}})
|\bx-\bx_n|^2 | \varphi_n(\bx)|,
\end{align*}
with $\boldsymbol{\eta_{x,x_n}} \in [\bx, \bx_n]$.
Because ${\mathcal R}_n$ are rotations, we have that
$\| \mathrm{Hess} (V_{{\mathcal R}_n})\|_2(\,\cdot\,)  = 
\|\mathrm{Hess} (V)\|_2({\mathcal R}_n\,\cdot\,)$
for $n \in \N$. Since
$\|\mathrm{Hess} (V)\|_2$  
varies with rate $\nu$ along $\mathrm{Im}(\gamma)$ and 
since, by \eqref{con4.3} and \eqref{crosscond}, 
$r_n/|\bx_n|^\nu \to 0$  as $n \to \infty$, 
we find a constant $C_{11}>0$ such that
for $n \in \N$ large enough
\begin{align*}
\| \mathrm{Hess} (V_{{\mathcal R}_n})\|_2(\boldsymbol{\eta})
\le C_{11} \| \mathrm{Hess} (V_{{\mathcal R}_n})\|_2(\bx_n) , \quad
 \boldsymbol{\eta} \in B_{2r_n}(\bx_n) 
\end{align*}
holds. As a consequence,
\begin{align*}
\big\| U_{{\mathcal R}_n} e^{-\ri g} 
\big[V_{{\mathcal R}_n} - V_n  -
{\mathcal E}_n(x_1-x_{1,n})\big] \varphi_n \big\|
&\le 4C_{11}r_n^2
\| \mathrm{Hess} (V)\|_2({\mathcal R}_n \bx_n) 
\|\varphi_n \| \\ & \le
4C_{11} \| \mathrm{Hess} (V)\|_2(\by_n) 
\left(\tfrac{|V_n|}{B_0} \right)^{1+\epsilon} 
\|\varphi_n \| 
\end{align*}
for $n \in \N$ large enough.
In view of \eqref{con4.2},  we conclude that 
$(U_{{\mathcal R}_n} e^{-\ri g}\varphi_n)_{n \in \N}$ 
is a Weyl sequence for $0$.
\end{proof}
%%%%%%%%%%%%%%%%%%%%%%%%%%%%%%%%%%%%%%%%%%%%%%%%%%%%%%%%%%%%%%%%%%%%%%%%%%%%%
%%%%%%%%%%%%%%%%%%%%%%%%%%%%%%%%%%%%%%%%%%%%%%%%%%%%%%%%%%%%%%%%%%%%%%%%%%%%%
%%%%%%%%%%%%%%%%%%%%%%%%%%%%%%%%%%%%%%%%%%%%%%%%%%%%%%%%%%%%%%%%%%%%%%%%%%%%%
%%%%%%%%%%%%%%%%%%%%%%%%%%%%%%%%%%%%%%%%%%%%%%%%%%%%%%%%%%%%%%%%%%%%%%%%%%%%%
%%%%%%%%%%%%%%%%%%%%%%%%%%%%%%%%%%%%%%%%%%%%%%%%%%%%%%%%%%%%%%%%%%%%%%%%%%%%%
\noindent
{\bf Acknowledgments.}
This work has been supported by SFB-TR12
"Symmetries and Universality in Mesoscopic Systems" of the DFG. 
The author also wants to thank Edgardo Stockmeyer for useful 
discussions and remarks as well as the
{\it Faculdad de F\'isica de la Pontificia Universidad Cat\'olica de Chile}
for the great hospitality.
%%%%%%%%%%%%%%%%%%%%%%%%%%%%%%%%%%%%%%%%%%%%%%%%%%%%%%%%%%%%%%%%%%%%%%%%%%%%%
\begin{appendix}
\section{Essential self-adjointness of the Pauli operator}\label{selfadj}
In this section we recapitulate an argument, originally given 
in \cite{Iwatsuka1990}, for proving the essential self-adjointness of
the Pauli operator. As we will see, this argumentation works also for 
the relaxed regularity conditions on $B$ and $V$ of Proposition \ref{EssSelf}.

For the proof we first note that for $\varphi \in C_0^\infty(\R^2,\C)$
we can write 
$$[(-\ri \nabla -{\bf A})^2 + B]\varphi = 
   \sum_{k,l = 1}^2 (-\ri \partial_k -A_k)\overline{C_{k,l}}(-\ri \partial_l -A_l)\varphi\,,$$
$$[(-\ri \nabla-{\bf A})^2 - B]\varphi = 
   \sum_{k,l = 1}^2 (-\ri \partial_k -A_k)C_{k,l}(-\ri \partial_l -A_l)\varphi\,,$$
where $C_{k,l}$ denote the coefficients of the symmetric non-negativ definite matrix
$$C  = \mathbbm{1}-\sigma_2  =
\begin{pmatrix}
1&\ri\\
-\ri & 1  
  \end{pmatrix} = C^*.$$
Furthermore, along the proof we use the notation
$B_R:= \{\bx \in \R^2 |\ |\bx| \le R\}$ and 
$S_R := \{\bx \in \R^2 |\ |\bx| = R\} $.
\begin{proof}[Proof of Proposition \ref{EssSelf}]
 Since $H$ is a diagonal matrix operator, it suffices to show that 
both operators on the diagonal,
$$Q_{\pm}:= [(-\ri \nabla -{\bf A})^2 \pm B +V],$$
are essentially self-adjoint on $C_0^\infty(\R^2,\C)$. Because
$Q_\pm$ are symmetric on $C_0^\infty(\R^2,\C) $,
we have to show $Q_+^*\varphi= \pm \ri \varphi $
implies $\varphi \equiv 0$ for $\varphi \in \mathcal{D}(Q_+^*)$,
respectively $Q_-^*\varphi= \pm \ri \varphi $
implies $\varphi \equiv 0$ for $\varphi \in \mathcal{D}(Q_-^*)$.
We only treat the case $Q_-^*\varphi= \ri \varphi $ since the 
others are completely analogous. 
Let $\varphi \in \mathcal{D}(Q_-^*)$ be such that 
$Q_-^*\varphi= \ri \varphi$, then
\begin{equation} \label{eq:dis} 
[(-\ri \nabla-{\bf A})^2 - B + V]\varphi = \ri \varphi
\end{equation}
holds in distributional sense. Due to elliptic regularity theory 
(see e.g. \cite{Gilbarg_Trudinger}, \cite{Hellwig}), we obtain 
that $\varphi \in C^2(\R^2,\C)$ and that \eqref{eq:dis} holds 
strongly. Applying integration by parts results in
\begin{align} \label{eq:dis7}
\begin{split}
 & \int_{B_R}\bigg[\sum_{k,l = 1}^2 (-\ri \partial_k -A_k)C_{k,l}
(-\ri \partial_l -A_l)\varphi\bigg] \bar \varphi\,  \rd^2x 
\\ & \hspace{2.63cm}
=  \int_{B_R}\bigg[\sum_{k,l = 1}^2 (-\ri \partial_k -A_k)\varphi C_{k,l} 
\overline{(-\ri \partial_l -A_l)\varphi}\bigg] \,  \rd^2x  
\\ &  \hspace{4.15cm}
+ \ri \int_{S_R}\bigg[\sum_{k,l = 1}^2 \bold \nu_k C_{k,l}
(-\ri \partial_l -A_l)\varphi\bigg] \bar \varphi\,  \rd S, 
\end{split}
\end{align}
with $R>0$, where  $\bold \nu_k(\bx) = x_k/|\bx|$ for $k= 1,2$.
By taking the imaginary part of \eqref{eq:dis7},
we conclude with \eqref{eq:dis} that
$$\int_{B_R} |\varphi|^2 \, \rd^2 x = 
\int_{S_R}\bigg[\sum_{k,l = 1}^2 \bold \nu_k C_{k,l}
(-\ri \partial_l -A_l)\varphi\bigg] \bar \varphi\,  \mbox{d}S$$
for any $R>0$. The Cauchy-Schwarz inequality yields
$$\int_{B_R} |\varphi|^2 \, \rd^2 x \le \!
\bigg(\int_{S_R}\sum_{k,l = 1}^2 (-\ri \partial_k -A_k)\varphi C_{k,l} 
\overline{(-\ri \partial_l -A_l)\varphi}\,\mbox{d}S\bigg)^{1/2} \!
\left(\int_{S_R} |\varphi|^2  \mbox{d}S\right)^{1/2}\!. $$
Hence, it suffices to show that 
\begin{equation} \label{eq:dis5}
\int_{\R^2}\sum_{k,l = 1}^2 \frac{(-\ri \partial_k -A_k)\varphi C_{k,l} 
\overline{(-\ri \partial_l -A_l)\varphi}}{|\bx|^2+1}\,\rd^2 x \ <\infty 
\end{equation}
since this implies that
$(1,\infty) \ni r  \mapsto r^{-1}\int_{B_r} |\varphi|^2 \, \mbox{d}\bx $
is an $L^1$-function, i.e. $\varphi \equiv 0$. \\ For \eqref{eq:dis5}
we consider the function
$$f(R) :=
\int_{B_R}\sum_{k,l = 1}^2 \frac{(-\ri \partial_k -A_k)\varphi C_{k,l} 
\overline{(-\ri \partial_l -A_l)\varphi}}{|\bx|^2+1} \,\rd^2 x\,, $$
with $R>0$. Using Equation \eqref{eq:dis} and integration by parts, we
obtain, with $\zeta(\bx) = ( |\bx|^2+1)^{-1}$ and $M \ge c+ |d|$, that
\begin{align*}
 f(R) - M\|\varphi\|^2   
&\le  f(R) + \int_{B_R}  \zeta V |\varphi|^2\,  \rd^2 x  \\
& =  \int_{B_R} \zeta (Q_-^*\varphi) \bar{\varphi}  \, \rd^2 x  -  
 \ri \int_{B_R} \bigg[\sum_{k,l = 1}^2 (\partial_l\zeta)
C_{k,l}(-\ri \partial_k -A_k)\varphi \bigg] \bar \varphi\, \,\rd^2x \\
&\hspace{3.19cm} +  \ri \int_{S_R} \zeta \bigg[\sum_{k,l = 1}^2 \bold\nu_l C_{k,l}
(-\ri \partial_k -A_k)\varphi
 \bigg] \bar \varphi\, \,\mbox{d}S .
\end{align*}
By the estimates
\begin{align*}
\bigg| \int_{B_R} \bigg[\sum_{k,l = 1}^2 (\partial_l\zeta)C_{k,l} 
& (-\ri \partial_k  -A_k)  \varphi \bigg] \bar \varphi\, \,\rd^2 x \bigg| \\
& \le \int_{B_R} 2\zeta^{1/2} \bigg|\sum_{k,l = 1}^2 
\bold \nu_l C_{k,l}(-\ri \partial_k -A_k)\varphi \bigg| |\varphi|\,\rd^2 x \\ 
& \le  2\int_{B_R}\bigg[\sum_{k,l = 1}^2 \zeta(-\ri \partial_k -A_k)\varphi C_{k,l} 
 \overline{(-\ri \partial_l -A_l)\varphi}\bigg]^{1/2} |\varphi| \,\rd^2 x \\
& \le  2[f(R)]^{1/2}\|\varphi\| \le \frac{1}{2}f(R) + 2 \|\varphi\|^2
\end{align*}
and
\begin{align*}
  \bigg|\int_{S_R}&\zeta\bigg[\sum_{k,l = 1}^2\bold\nu_k C_{k,l} 
(-\ri \partial_l - A_l)\varphi\bigg] \bar \varphi\,  \rd S \bigg| 
\\ & \le 
\int_{S_R}\zeta\bigg[\sum_{k,l = 1}^2 (-\ri \partial_k -A_k)\varphi C_{k,l} 
\overline{(-\ri \partial_l -A_l)\varphi}\bigg]^{1/2} |\varphi|\, \rd S 
\\ & \le
\bigg(\int_{S_R}\sum_{k,l = 1}^2 \zeta(-\ri \partial_k -A_k)\varphi C_{k,l} 
\overline{(-\ri \partial_l -A_l)\varphi}\,\rd S\bigg)^{1/2} \!
\left(\int_{S_R} |\varphi|^2  \rd S\right)^{1/2} \\ & =
\left({f'(R)} \,\int_{S_R} |\varphi|^2  \rd S\right)^{1/2} ,
\end{align*}
we conclude that
$$f(R) \le 2(3+c) \|\varphi\|^2+ 
2\left({f'(R)} \,\int_{S_R} |\varphi|^2\rd S\right)^{1/2}\!.  
$$
% \begin{align*}
%  f(R)\  
% &\le \ \int_{B_R} |\varphi|^2  \mbox{d}\bx \ + \ 
%  \bigg|\int_{S_R}\bigg[\sum_{k,l = 1}^2 \bold \nu_k C_{k,l}
% (-\ri \partial_l -A_l)\varphi\bigg] \bar \varphi\,  \mbox{d}S \bigg| \\
% & \le \ \int_{\R^2} |\varphi|^2  \mbox{d}\bx \ + \ 
% \int_{S_R}\bigg[\sum_{k,l = 1}^2 (-\ri \partial_k -A_k)\varphi C_{k,l} 
% \overline{(-\ri \partial_l -A_l)\varphi}\bigg]^{1/2} |\varphi| \,\mbox{d}S \\
% & \le \ \|\varphi\|^2 \ + \ 
% \bigg(\int_{S_R}\sum_{k,l = 1}^2 (-\ri \partial_k -A_k)\varphi C_{k,l} 
% \overline{(-\ri \partial_l -A_l)\varphi}\,\mbox{d}S\bigg)^{1/2} 
% \left(\int_{S_R} |\varphi|^2  \mbox{d}S\right)^{1/2} \\
% & =\ \|\varphi\|^2 \ + \ 
% \left({f'(R)} \,\int_{S_R} |\varphi|^2  \mbox{d}S\right)^{1/2}.
% \end{align*}
If $f(R)= 0$ for all $R>0$, then clearly \eqref{eq:dis5} holds. If 
$f(R_0) > 0$ for some $R_0>0$, then $f(R) > 0$ for all $R>R_0$ and 
$f'(R)/f^2(R) \in L^1((R_0, \infty))$, implying that
$$\left(\frac{f'(R)}{f^2(R)} \,
\int_{S_R} |\varphi|^2  \mbox{d}S\right)^{1/2} \ \in L^1((R_0, \infty)).$$
Hence, there exists a sequence $(R_n)_{n\in \N} \subset (0,\infty)$ such that
$R_n \to \infty$ as $n \to \infty$ and
$$\left({f'(R_n)}\,
\int_{S_{R_n}} |\varphi|^2  \mbox{d}S\right)^{1/2}  \le \frac{1}{4} \, f(R_n).$$
Therefore, we have $f(R_n) \le 4(3+c) \|\varphi\|^2$ for all $n\in\N$, 
which implies  
\eqref{eq:dis5} since $f(R)$ is a monotonically increasing function.
\end{proof}
%%%%%%%%%%%%%%%%%%%%%%%%%%%%%%%%%%%%%%%%%%%%%%%%%%%%%%%%%%%%%%%%%%%%%%%%%%%%%%%%
%%%%%%%%%%%%%%%%%%%%%%%%%%%%%%%%%%%%%%%%%%%%%%%%%%%%%%%%%%%%%%%%%%%%%%%%%%%%%%%%
%%%%%%%%%%%%%%%%%%%%%%%%%%%%%%%%%%%%%%%%%%%%%%%%%%%%%%%%%%%%%%%%%%%%%%%%%%%%%%%%
\section{Remarks on locally compact operators}
\begin{lemma}\label{Gertrud}
Let $A$ be a locally compact, self-adjoint operator on $L^2(\R^n,\C^m)$ 
with $n,m \ge 1$. Assume there is a function
$W \in L_{loc}^\infty(\R^n, [0,\infty))$, 
with $W(\bx) \to \infty$ as $|\bx| \to\infty$,
such that
$$\| A \varphi\| \ge \|W\varphi\| \quad \mbox{for} \ 
\varphi \in \mathcal{D}(A). $$
Then $\sigma_{{\rm ess}}(A) = \emptyset$, i.e.
$A$ has only discrete spectrum.
\end{lemma}
\begin{proof}
Assume $\lambda \in \sigma_{{\rm ess}}(A) \subset \R$. Then there is 
a normalized sequence $(\varphi_n)_{n \in\N} \subset \mathcal{D}(A)$
such that
$\varphi_n \rightharpoonup 0$ 
and $\|(A-\lambda)\varphi_n\| \to 0$
as $n \to \infty$. Let $R> 0$ be a fixed constant.  We have 
\begin{align*}
\| \chi_R \varphi_n \| &=
\| \chi_R (A-\ri - \lambda)^{-1}  (A-\ri - \lambda) \varphi_n \| 
\\ & \le
 \| \chi_R (A-\ri - \lambda)^{-1}\| 
 \| (A- \lambda) \varphi_n \| +
 \| \chi_R (A-\ri - \lambda)^{-1} \varphi_n \| \,,
\end{align*}
using the notation $\chi_R := \chi_{B_R(0)}$. Since 
$ \chi_R (A-\ri -\lambda)^{-1}$ is compact, this inequality implies 
that $\|\chi_R\varphi_n\|\to 0$ as $n \to \infty$. Let $ R>0$ be so 
large that $ W(\bx) \ge 5|\lambda| + 1 $ if $|\bx| \ge R$.
Choosing $N \in \N$ large enough,  we can estimate, for $n \ge N$,
\begin{align*}
\| (A-\lambda) \varphi_n\|  & \ge
\|W\varphi_n\| - |\lambda| \\ & \ge
\|W(\mathbbm{1}-\chi_R) \varphi_n\| - 
\|W\chi_R\|_\infty \|\chi_R\varphi_n\| -|\lambda| \\ & \ge
(5|\lambda|+1)\|(\mathbbm{1}-\chi_R)\varphi_n\|-
\|W\chi_R\|_\infty \|\chi_R\varphi_n\| -|\lambda| \\ & \ge
(|\lambda| +1/2) - \|W\chi_R\|_\infty \|\chi_R\varphi_n\|\,. 
\end{align*}
Hence, $\|(A-\lambda)\varphi_n\| \nrightarrow 0$ as $n \to \infty$,
which is a contradiction.
\end{proof}
%%%%%%%%%%%%%%%%%%%%%%%%%%%%%%%%%%%%%%%%%%%%%%%%%%%%%%%%%%%%%%%%%%%%%%%%%%%%%
\section{Integral estimates}
\begin{lemma}\label{hermiteestimate}
 Let $\vartheta_n$ be the $n-$th Hermite polynomial. Then, 
for any $\epsilon >0$, it holds that 
\begin{align*}
\frac{1}{2^n n!}\int_{\sqrt{n^{1+\epsilon}}}^\infty
|\vartheta_n (x)|^2 e^{-x^2}  \rd x \ &\longrightarrow 0 
\quad \mbox{as} \ n \to \infty,\\[0.2cm]
\frac{1}{2^n n!}\int_{-\infty}^{-\sqrt{n^{1+\epsilon}}}
|\vartheta_n (x)|^2 e^{-x^2}  \rd x \ &\longrightarrow 0 
\quad \mbox{as} \ n \to \infty.
\end{align*} 
\end{lemma}
\begin{proof}
We only treat the first case since the second claim can 
be deduced from the first one by a symmetry argument.
Due to the identity 
$$\vartheta_n(x) = (-1)^n \sum_{k_1+2 k_2=n} 
\frac{n!}{k_1!k_2!} (-1)^{k_1+k_2}(2x)^{k_1}$$
for the $n-$th Hermite polynomial (see e.g. \cite{Abramowitz}),
we obtain for $|x|\ge 1$ the estimate
\begin{align*}
|\vartheta_n(x)| \le \sum_{k_1+2 k_2=n} 
\frac{n!}{k_1!k_2!} (2|x|)^{k_1}
\le \frac{n+1}{2}\,  2^n n! |x|^n.
\end{align*}
Thus, for $n \in \N$ large enough we have
\begin{align*}
\frac{1}{2^nn!}\int_{\sqrt{n^{1+\epsilon}}}^\infty
|\vartheta_n(x)|^2 e^{-x^2}\rd x 
&\le \frac{(n+1)^2}{4}2^n n!\int_{\sqrt{n^{1+\epsilon}}}^\infty
x^{2n} e^{-x^2} \rd x \\[0.1cm]
&\le \frac{(n+1)^2}{4}2^n
n^{1+\epsilon} n! \, \exp{(-n^{1+\epsilon}/2)}
\int_{\sqrt{n^{1+\epsilon}}}^\infty e^{-x^2/2}  \rd x, 
\end{align*}
and the r.h.s. tends to $0$ as $n \to \infty$ by Stirling's 
Formula.
\end{proof}
\end{appendix}
%%%%%%%%%%%%%%%%%%%%%%
%%%%%%%%%%%%%%%%%%%%%
%%%%%%%%%%%%%%%%%%%%%
\bibliographystyle{plain} 
%\bibliography{pauli}

\end{document}